\documentclass[11pt,twoside]{atmp}

\usepackage{amssymb}
\usepackage{amsmath}
\usepackage{amsthm}

\title[Covariant Hamiltonians]
{Diffeomorphism-invariant \\
Covariant Hamiltonians \\
of a pseudo-Riemannian Metric \\
and a Linear Connection }

\author[J. Mu\~noz Masqu\'e, M. E. Eugenia Rosado]
{J. Mu\~noz\ Masqu\'e$^\dag $,
M. Eugenia Rosado Mar\'{\i}a$^\ddag $}

\address{${}^\dag $Instituto de F\'{\i}sica Aplicada, CSIC \\
C/ Serrano 144, 28006-Madrid, Spain\\
${}^\ddag $Departamento de Matem\'atica Aplicada \\
Escuela T\'ecnica Superior de Arquitectura, UPM \\
Avda.\ Juan de Herrera 4,\ 28040-Madrid, Spain}

\addressemail{jaime@iec.csic.es, eugenia.rosado@upm.es}

\date{}

\newtheorem{theorem}{Theorem}[section]
\newtheorem{proposition}[theorem]{Proposition}
\newtheorem{lemma}[theorem]{Lemma}

\theoremstyle{remark}
\newtheorem{remark}[theorem]{Remark}

\begin{document}

\arxurl{ }

\begin{abstract}
\noindent Let $M\to N$ (resp.\ $C\to N$) be
the fibre bundle of pseudo-Riemannian metrics
of a given signature (resp.\ the bundle of linear
connections) on an orientable connected manifold
$N$. A geometrically defined class of first-order
Ehresmann connections on the product fibre bundle
$M\times _NC$ is determined such that, for every
connection $\gamma $ belonging to this class
and every $\mathrm{Diff}N$-invariant Lagrangian
density $\Lambda $ on $J^1(M\times _NC)$,
the corresponding covariant Hamiltonian
$\Lambda ^\gamma $ is also $\mathrm{Diff}N$-invariant.
The case of $\mathrm{Diff}N$-invariant second-order
Lagrangian densities on $J^2M$ is also studied
and the results obtained are then applied
to Palatini and Einstein-Hilbert Lagrangians.
\end{abstract}

\maketitle

\cutpage 

\setcounter{page}{2}


\medskip


\medskip


\medskip

\noindent\textit{Acknowledgements:\/} Supported
by Ministerio de Ciencia e Innovaci\'on of Spain,
under grant \#MTM2008--01386.

\noindent
\section{Introduction}

In Mechanics, the Hamiltonian function attached
to a Lagrangian density
$\Lambda =L(t,q^i,\dot{q}^i)dt$ on $\mathbb{R}\times TQ$
is given by $H=\dot{q}^i\partial L/\partial \dot{q}^i-L$,
but---as it was early observed in \cite{GoS}---this is not
an invariant definition if an arbitrary fibred manifold
$t\colon E\to\mathbb{R}$ is considered
(thus generalizing the notion of an absolute time)
instead of the direct product bundle
$\mathbb{R}\times Q\to\mathbb{R}$; e.g., see \cite{CMS},
\cite{MSarda}, \cite{MassaPagani} for this point of view.
In this case, an Ehresmann connection is needed
in order to lift the vector field $\partial /\partial t$
from $\mathbb{R}$ to $E$, and the Hamiltonian is then
defined by applying the Poincar\'e-Cartan form attached to
$\Lambda $ to the horizontal lift of $\partial /\partial t$.

In the field theory---where no distinguished vector
field exists on the base manifold---the need of
an Ehresmann connection is even greater, in order to attach
a covariant Hamiltonian to each Lagrangian
density; e.g., see \cite[4.1]{MSh}, \cite{MSarda}, and
the definitions below.

Let $p\colon E\to N$ be an arbitrary fibred manifold
over a connected manifold $N$, $n=\dim N$, $\dim E=m+n$,
oriented by $v_n=dx^1\wedge \cdots \wedge dx^n$.
Throughout this paper, Latin (resp.\ Greek) indices run from
$1$ to $n$ (resp.\ $m$). An Ehresmann connection on
a fibred manifold $p\colon E\to N$ is a differential
$1$-form $\gamma$ on $E$ taking values in the vertical
sub-bundle $V(p)$ such that $\gamma(X)=X$ for every
$X\in  V(p)$ (e.g., see \cite{MSarda}, \cite{MSh},
\cite{Sarda}, \cite{Saunders}). Once an Ehresmann connection
$\gamma $ is given, a decomposition of vector bundles holds
$T(E)=V(p)\oplus \ker \gamma $, where $\ker \gamma $ is called
the horizontal sub-bundle determined by $\gamma $.
In a fibred coordinate system $(x^j,y^\alpha )$ for $p$,
an Ehresmann connection can be written as
\[
\gamma =(dy^\alpha +\gamma _j^\alpha dx^j)\otimes
\frac{\partial }{\partial y^\alpha },
\quad
\gamma _j^\alpha \in  C^\infty (E).
\]
According to \cite{MSh}, the covariant Hamiltonian $\Lambda ^\gamma $
associated to a Lagrangian density on $J^1E$, $\Lambda =Lv_n$,
$L\in  C^\infty (J^1E)$, with respect to $\gamma $ is the Lagrangian
density defined by,
\begin{equation}
\Lambda ^\gamma
=\left(
(p_0^1)^\ast \gamma -\theta
\right)
\wedge \omega _\Lambda -\Lambda,
\label{Lambdagamma}
\end{equation}
where, $p_0^1\colon J^1E\to J^0E=E$ is the projection mapping,
$\theta =\theta ^\alpha \otimes \partial /\partial y^\alpha $,
$\theta ^\alpha =dy^\alpha -y_i^\alpha dx^i $ is the $V(p)$-valued
$1$-form on $J^1E$ associated with the contact structure,
written on a fibred coordinate system $(x^i,y^\alpha )$,
and $\omega _\Lambda $ is the Legendre form attached
to $\Lambda $, i.e., the $V^\ast (p)$-valued $p^1$-horizontal
$(n-1)$-form on $J^1E$ given by
\[
\omega_\Lambda
=(-1)^{i-1}
\frac{\partial L}{\partial y_i^\alpha }i_{\partial/\partial x^i}
v_n\otimes dy^\alpha ,
\]
where $(x^i,y^\alpha ;y_{i}^\alpha )$ is the coordinate
system induced from $(x^i,y^\alpha )$ on the $1$-jet bundle
and $p^1\colon J^1E\to N$ is the projection
on the base manifold. Locally,
\begin{equation}
\Lambda^\gamma
=\Bigl(
\left(
\gamma _i^\alpha +y_i^\alpha
\right)
\frac{\partial L}{\partial y_i^\alpha }-L\Bigr)
dx^1\wedge\cdots\wedge dx^{n}.
\label{Lambdagamma1}
\end{equation}
From \eqref{Lambdagamma} we obtain the following decomposition
of the Poincar\'{e}-Cartan form attached to $\Lambda $
(e.g., see \cite{Gotay}, \cite{MSarda}, \cite{MP}):
$\Theta_\Lambda=\theta\wedge\omega_\Lambda +\Lambda
=(p_0^1)^\ast \gamma\wedge\omega_\Lambda-\Lambda^\gamma $.

A diffeomorphism $\Phi\colon E\to E$ is said to be
an automorphism of $p$ if there exists
$\phi\in \mathrm{Diff}N$ such that
$p\circ\Phi=\phi\circ p$. The set of such automorphisms
is denoted by $\mathrm{Aut}(p)$ and its Lie
algebra is identified to the space
$\mathrm{aut}(p)\subset\mathfrak{X}(E)$ of $p$-projectable
vector fields on $E$. Given a subgroup
$\mathcal{G}\subseteq\mathrm{Aut}(p)$,
a Lagrangian density $\Lambda $ is said to be
$\mathcal{G}$-invariant if $(\Phi^{(1)})^\ast \Lambda=\Lambda $
for every $\Phi\in \mathcal{G}$, where
$\Phi^{(1)}\colon J^1E\to J^1E$ denotes the
$1$-jet prolongation of $\Phi$. Infinitesimally,
the $\mathcal{G}$-invariance equation can be reformulated as
$L_{X^{(1)}}\Lambda=0$ for every
$X\in \mathrm{Lie}(\mathcal{G})$, $X^{(1)}$ denoting the
$1$-jet prolongation of the vector field $X$.

When a group $\mathcal{G}$ of transformations of $E$ is given,
a natural question arises:

\begin{itemize}
\item Determine a class---as small as possible---
of Ehresmann connections $\gamma$ such that
$\Lambda^\gamma $ is $\mathcal{G}$-invariant for every
$\mathcal{G}$-invariant Lagrangian density $\Lambda $.
\end{itemize}

Below we tackle this question in the framework
of General Relativity, i.e., the group $\mathcal{G}$
is the group of all diffeomorphisms of the ground
manifold $N$ acting in a natural way either on
the bundle of pseudo-Riemannian metrics
$p_M\colon M=M(N)\to N$ of a given signature
$(n^{+},n^{-})$, $n^{+}+n^{-}=n$, or on the product bundle
$p\colon M\times _NC\to N$, where $p_C\colon C=C(N)\to N$
is the bundle of linear connections on $N$.
Namely, we solve the following two problems:

\begin{description}
\item[(P)] Determine a class---as small as possible---
of Ehresmann connections $\gamma$ such that for every
$\mathrm{Diff}N$-invariant first-order Lagrangian
density $\Lambda $ on the bundle $J^1(M\times _NC)$ ,
the corresponding covariant Hamiltonian
$\Lambda^\gamma $ is also $\mathrm{Diff}N$-invariant.
\end{description}

Similarly to the problem \textbf{(P)},
we formulate the corresponding problem on $J^2M$ as follows:

\begin{description}
\item[(P2)] Determine a class of second-order Ehresmann
connections $\gamma^2$ on $M$ such that for every
$\mathrm{Diff}N$-invariant second-order Lagrangian density
$\Lambda $ on the bundle $J^2M$, the corresponding
covariant Hamiltonian $\Lambda^{\gamma^2}$---defined in
\eqref{Lambda_gamma_2}---is also $\mathrm{Diff}N$-invariant.
\end{description}
\noindent Essentially, a class of first-order Ehresmann connections
on the bundle $M\times _NC$ is obtained, defined by the conditions
$(C_M)$ and $(C_C)$ below (see Propositions \ref{proposition0}
and \ref{proposition3}), solving the problem \textbf{(P)}. This class
of connections also helps to solve \textbf{(P2)} by means of a natural
isomorphism between $J^1M$ and $M\times _NC^{\mathrm{sym}}$, where
$C^{\mathrm{sym}}$ denotes the sub-bundle of symmetric connections
on $N$ (cf.\ Theorem \ref{one-to-one}). Finally, this approach
is applied to Palatini and Einstein-Hilbert Lagrangians (\cite{BurtonMann1},
\cite{BurtonMann2}), obtaining results compatible with their usual
Hamiltonian formalisms.

\section{Invariance under diffeomorphisms}

\subsection{Preliminaries}

\subsubsection{Jet-bundle notations}

Let $p^k\colon J^kE\to N$ be the $k$-jet bundle of local sections
of an arbitrary fibred manifold $p\colon E\to N$, with
projections $p_l^k\colon J^kE$ $\to J^lE$,
$p_l^k(j_x^ks)=j_x^ls$, for $k\geq l$, $j_x^ks$ denoting the $k$-jet at
$x$ of a section $s$ of $p$ defined on a neighbourhood of $x\in  N$.

A fibred coordinate system $(x^i,y^\alpha )$ on $V$ induces a coordinate
system $(x^i,y_{I}^\alpha )$, $I=(i_1,\dotsc,i_N)\in \mathbb{N}^n$,
$0\leq|I|=i_{1}+\cdots+i_N\leq r$, on $(p_0^r)^{-1}(V)=J^rV$ as
follows: $y_I^\alpha (j_x^r s)
=(\partial^{|I|}(y^\alpha \circ s)/\partial x^I)(x)$,
with $y_0^\alpha =y^\alpha $.

Every morphism $\Phi\colon E\to E^\prime $ whose associated map
$\phi\colon N\to N^\prime $ is a diffeomorphism, induces a map
\begin{equation}
\begin{array}
[c]{l}
\Phi^{(r)}\colon J^rE\to J^rE^\prime ,\\
\Phi^{(r)}(j_x^rs)=j_{\phi(x)}^{r}(\Phi\circ s\circ\phi^{-1}).
\end{array}
\label{Phi^r}
\end{equation}
If $\Phi _t$ is the flow of a vector field $X\in \mathrm{aut}(p)$,
then $\Phi _t^{(r)}$ is the flow of a vector field
$X^{(r)}\in \mathfrak{X}(J^rE)$, called the infinitesimal
contact transformation of order $r$
associated to the vector field $X$. The mapping
\[
\mathrm{aut}(p)\ni X\mapsto X^{(r)}\in \mathfrak{X}(J^rE)
\]
is an injection of Lie algebras, namely, one has
\[
\begin{array}
[c]{l}
(\lambda X+\mu Y)^{(r)}=\lambda X^{(r)}+\mu Y^{(r)},\smallskip \\
\lbrack X,Y]^{(r)}=[X^{(r)},Y^{(r)}], \smallskip \\
\forall \lambda ,\mu \in \mathbb{R},\;\forall X,Y\in \mathrm{aut}(p).
\end{array}
\]
In particular, for $r=1$,
\begin{align*}
X  & =u^i\frac{\partial}{\partial x^i}
+v^\alpha \frac{\partial}{\partial y^\alpha },
\quad
u^i\in  C^\infty (N),v^\alpha \in  C^\infty (E),\\
X^{(1)}
& =u^i\frac{\partial}{\partial x^i}
+v^\alpha \frac{\partial }{\partial y^\alpha }
+v_{i}^\alpha \frac{\partial}{\partial y_{i}^\alpha },
\quad
v_{i}^\alpha =\frac{\partial v^\alpha }{\partial x^i}
+y_{i}^{\beta}\frac{\partial v^\alpha }{\partial y^{\beta}}
-y_{k}^\alpha \frac{\partial u^k}{\partial x^i}.
\end{align*}

\subsubsection{Coordinates on $M(N)$, $F(N)$,
$C(N)$}\label{Coordinate Systems}

Every coordinate system $(x^i)$ on an open domain $U\subseteq N$
induces the following coordinate systems:

\begin{enumerate}
\item $(x^i,y_{jk})$ on $(p_M)^{-1}(U)$, where $p_M\colon M \to N$
is the bundle of metrics of a given signature, and the functions
$y_{jk}=y_{kj}$ are defined by,
\begin{equation}
g_x=\sum _{i\leq j}y_{ij}(g_x)(dx^i)_x\otimes(dx^j)_x,\;
\forall
g_x\in (p_M)^{-1}(U).
\label{inducedsystemM(N)}
\end{equation}

\item $(x^i,x_j^i)$ on $(p_F)^{-1}(U)$, where
$p_F\colon F(N)\to N$ is the bundle of linear frames
on $N$, and the functions $x_j^i$ are defined by,
\[
u=\left(
(\partial /\partial x^1)_x,\dotsc,(\partial/\partial x^n)_x
\right)
\cdot
\left(  x_j^i(u)
\right) ,
\;x=p_F(u),\forall u\in (p_F)^{-1}(U),
\]
or equivalently,
\begin{equation}
u=(X_1,\dotsc,X_N)\in  F_x(N),\;X_j=x_j^i(u)
\left(
\frac{\partial }{\partial x^i}
\right) _x,\quad1\leq j\leq n.
\label{frame}
\end{equation}

\item $(x^i,A_{kl}^j)$ on $(p_C)^{-1}(U)$, where $p_C\colon C\to N$
is the bundle of linear connections on $N$, and the functions
$A_{kl}^j$ are defined as follows. We first recall some basic facts.
Connections on $F(N)$ (i.e., linear connections of $N$)
are the splittings of the Atiyah sequence (cf.\ \cite{Bruzzo}),
\[
0\to\mathrm{ad}F(N)\to T_{Gl(n,\mathbb{R})}F(N)
\overset{(p_{F})_\ast }{-\!\!\!\!-\!\!\!\longrightarrow}
TN \to 0,
\]
where

\begin{enumerate}
\item $\mathrm{ad}F(N)=T^\ast N\otimes TN$ is the adjoint bundle,

\item $T_{Gl(n,\mathbb{R})}(F(N))=T(F(N))/Gl(n,\mathbb{R})$, and

\item $\mathrm{gau}F(N)=\Gamma(N,\mathrm{ad}F(N))$
is the gauge algebra of $F(N)$.
\end{enumerate}
\noindent
We think of $\mathrm{gau}F(N)$ as the `Lie algebra'
of the gauge group $\mathrm{Gau}F(N)$. Moreover,
$p_C\colon C\to N$ is an affine bundle modelled
over the vector bundle $\otimes^2T^\ast N\otimes TN$.
The section of $p_C$ induced tautologically by the
linear connection $\Gamma$ is denoted by
$s_\Gamma\colon N\to C$. Every $B\in \mathfrak{gl}(n,\mathbb{R})$
defines a one-parameter group
$\varphi_t^B\colon U\times Gl(n,\mathbb{R})\to U\times Gl(n,\mathbb{R})$
of gauge transformations by setting (cf.\ \cite{CM}),
$\varphi_t^{B}(x,\Lambda)=(x,\exp(tB)\cdot\Lambda)$.
Let us denote by $\bar{B}\in \mathrm{gau}(p_F)^{-1}(U)$
the corresponding infinitesimal generator. If $(E_j^i)$
is the standard basis of $\mathfrak{gl}(n,\mathbb{R})$, then
$\bar{E}_j^i=\sum_{h=1}^nx_h^j\partial/\partial x_h^i$,
for $i,j=1,\dotsc,n$, is a basis of $\mathrm{gau}(p_F)^{-1}(U)$.
Let $\tilde{E}_j^i=\bar{E}_j^i\operatorname{mod}G$ be the class
of $\bar{E}_j^i$ on $\mathrm{ad}F(N)$. Unique smooth functions
$A_{jk}^i$ on $(p_C)^{-1}(U)$ exist such that,
\begin{align}
s_\Gamma
\left(
\frac{\partial}{\partial x^j}
\right)
& =\frac{\partial }{\partial x^j}
-(A_{jk}^i\circ\Gamma)\tilde{E}_k^i
\label{inducedsystemC(N)}\\
& =\frac{\partial}{\partial x^j}
-(A_{jk}^i\circ\Gamma)x_h^k
\frac{\partial}{\partial x_h^i},
\nonumber
\end{align}
for every $s_\Gamma$ and
$A_{jk}^i(\Gamma _x)=\Gamma _{jk}^i(x)$, where
$\Gamma _{jk}^i$ are the Christoffel symbols
of the linear connection $\Gamma $ in the coordinate
system $(x^i)$, see \cite[III, Poposition 7.4]{KN}.
\end{enumerate}

\subsection{Natural lifts\label{natural_lifts}}

Let $f_M\colon M\to M$, cf.\ \cite{MV} (resp.\
$\tilde{f}\colon F(N)\to F(N)$, cf.\ \cite[p. 226]{KN})
be the natural lift of $f\in \mathrm{Diff}N$
to the bundle of metrics (resp.\ linear frame bundle);
namely $f_M(g_x)=(f^{-1})^\ast g_x$ (resp.\
$\tilde{f}(X_1,\dotsc,X_N)=(f_\ast X_1,\dotsc,f_\ast X_N)$,
where $(X_1,\dotsc,X_N)\in  F_x(N)$); hence
$p_M\circ f_M=f\circ p_M$
(resp.\ $p_F\circ\tilde{f}=f\circ p_F$), and $f_M\colon M\to M$
(resp.\ $\tilde{f}\colon F(N)\to F(N)$) have a natural extension to jet
bundles $f_M^{(r)}\colon J^r(M)\to J^r(M)$ (resp.\
$\tilde {f}^{(r)}\colon J^r(FN)\to J^r(FN)$) as defined in the formula
\eqref{Phi^r}, i.e.,
\[
f_M^{(r)}
\left(
j_x^rg
\right)
=j_{f(x)}^r(f_M\circ g\circ f^{-1})\quad
\text{(resp.\ }\tilde{f}^{(r)}
\left(
j_x^rs
\right)
=j_{f(x)}^r(\tilde{f}\circ s\circ f^{-1})
\text{)}.
\]
As $\tilde{f}$ is an automorphism of the principal
$Gl(n,\mathbb{R})$-bundle $F(N)$, it acts on linear connections
by pulling back connection forms, i.e.,
$\Gamma^\prime =\tilde{f}
\left(
\Gamma
\right) $ where
$\omega _{\Gamma^\prime }=(\tilde{f}^{-1})^\ast \omega_{\Gamma}$
(see\ \cite[II, Proposition\ 6.2-(b)]{KN}, \cite[3.3]{CM}).
Hence, there exists a unique diffeomorphism
$\tilde{f}_C\colon C\to C$ such that,

1) $p_C\circ\tilde{f}_C=f\circ p_C$, and

2) $\tilde{f}_C\circ s_{\Gamma}=s_{\tilde{f}\left(  \Gamma\right)  }$
for every linear connection $\Gamma$.

If $f_t$ is the flow of a vector field $X\in \mathfrak{X}(N)$, then the
infinitesimal generator of $(f_{t})_M$ (resp.\ $\tilde{f}_t$,
resp.\ $(\tilde{f}_{t})_C$) in $\mathrm{Diff}M$ (resp.\ $\mathrm{Diff}F(N)$,
resp.\ $\mathrm{Diff}C$) is denoted by $X_M$ (resp.\ $\tilde{X}$, resp.\
$\tilde{X}_C$) and the following Lie-algebra homomorphisms are obtained:
\[
\left\{
\begin{array}
[c]{ll}
\mathfrak{X}(N)\to\mathfrak{X}(M),
& X\mapsto X_M\\
\mathfrak{X}(N)\to\mathfrak{X}(F(N)),
& X\mapsto \tilde{X}\\
\mathfrak{X}(N)\to \mathfrak{X}(C),
& X\mapsto \tilde{X}_C
\end{array}
\right.
\]

If $X=u^i\partial/\partial x^i\in \mathfrak{X}(N)$
is the local expression for $X$, then

\begin{enumerate}
\item From \cite[eqs. (2)--(4)]{MV} we know that the natural lift
of $X$ to $M$ is given by,
\[
X_M=u^i\frac{\partial}{\partial x^i}-\sum_{i\leq j}
\left(
\frac{\partial u^h}{\partial x^i}y_{hj}
+\frac{\partial u^h}{\partial x^j}y_{ih}
\right)
\frac{\partial}{\partial y_{ij}}\in \mathfrak{X}(M).
\]
and its $1$-jet prolongation,
\begin{multline*}
\!\!\!\!\!\!\!X_M^{(1)}=u^i\frac{\partial}{\partial x^i}
-\sum_{i\leq j}
\left(
\!\frac{\partial u^h}{\partial x^i}y_{hj}
+\frac{\partial u^h}{\partial x^j}y_{hi}\!
\right)
\frac{\partial}{\partial y_{ij}}\\
\!\! -\sum_{i\leq j}\!
\left(
\!\frac{\partial^2u^h}{\partial x^i\partial x^k}y_{hj}
+\frac{\partial^2u^h}{\partial x^j\partial x^k}y_{hi}
+\frac{\partial u^h}{\partial x^i}y_{hj,k}
+\frac{\partial u^h}{\partial x^j}y_{hi,k}
+\frac{\partial u^h}{\partial x^k}y_{ij,h}
\!
\right)
\!\frac{\partial}{\partial y_{ij,k}}.
\end{multline*}

\item From \cite[Proposition 3]{EM} (also see \cite[VI, Proposition 21.1]{KN})
we know that the natural lift of $X$ to $F(N)$ is given by,
\[
\tilde{X}=u^i\frac{\partial}{\partial x^i}
+\frac{\partial u^i}{\partial x^l}x_j^l
\frac{\partial}{\partial x_j^i},
\]
and its $1$-jet prolongation,
\begin{align*}
\tilde{X}^{(1)}
& =u^i\frac{\partial}{\partial x^i}
+\frac{\partial u^i}{\partial x^l}x_j^l
\frac{\partial}{\partial x_j^i}
+v_{jk}^i\frac{\partial}{\partial x_{j,k}^i},\\
v_{jk}^i
& =\frac{\partial u^i}{\partial x^l}x_{j,k}^l
-\frac{\partial u^l}{\partial x^k}x_{j,l}^i
+\frac{\partial^2u^i}{\partial x^k\partial x^l}x_j^l.
\end{align*}

\item Finally,
\[
\tilde{X}_C=u^i\frac{\partial}{\partial x^i}
-\left(
\frac{\partial ^2u^i}{\partial x^j\partial x^k}
-\frac{\partial u^i}{\partial x^l}A_{jk}^l
+\frac{\partial u^l}{\partial x^k}A_{jl}^i
+\frac{\partial u^l}{\partial x^j}A_{lk}^i
\right)  \frac{\partial}{\partial A_{jk}^i},
\]

\end{enumerate}

\begin{align}
\tilde{X}_C^{(1)}
& =u^i\frac{\partial}{\partial x^i}+w_{jk}^i
\frac{\partial}{\partial A_{jk}^i}+w_{jkh}^i
\frac{\partial}{\partial A_{jk,h}^i},
\nonumber\\
w_{jk}^i
& =-\frac{\partial^2u^i}{\partial x^j\partial x^k}
+\frac{\partial u^i}{\partial x^l}A_{jk}^l
-\frac{\partial u^l}{\partial x^k}A_{jl}^i
-\frac{\partial u^l}{\partial x^j}A_{lk}^i,
\label{3_2}\\
w_{jkh}^i
& =-\frac{\partial^3u^i}
{\partial x^h\partial x^j\partial x^k}
+\frac{\partial^2u^i}{\partial x^h\partial x^l}A_{jk}^l
-\frac{\partial^2u^l}{\partial x^h\partial x^k}A_{jl}^i
-\frac{\partial^2u^l}{\partial x^h\partial x^j}A_{lk}^i
\label{3_3}\\
& +\frac{\partial u^i}{\partial x^l}A_{jk,h}^l
-\frac{\partial u^l}{\partial x^k}A_{jl,h}^i
-\frac{\partial u^l}{\partial x^j}
A_{lk,h}^i-\frac{\partial u^l}{\partial x^h}A_{jk,l}^i.\nonumber
\end{align}

Let $p\colon M\times _NC\to N$ be the natural projection.

We denote by $\bar{f}=(f_M,\tilde{f}_C)$ (resp.\
$\bar{X}=(X_{M},\tilde{X}_C)\in \mathfrak{X}(M\times _NC)$)
the natural lift of $f$ (resp.\ $X$)
to $M\times _NC$. The prolongation to the bundle
$J^1(M\times _NC)$ of $\bar{X}$ is as follows:
\begin{align}
\quad\bar{X}^{(1)}\!\!
& =\!\!
\left(
\!X_{M}^{(1)},\tilde{X}_C^{(1)}\!
\right)
\label{X_bar^1}\\
\!\!  & =\!\!u^i\frac{\partial}{\partial x^i}
+\sum_{i\leq j}v_{ij}
\frac{\partial}{\partial y_{ij}}
+\sum_{i\leq j}v_{ijk}\frac{\partial}{\partial y_{ij,k}}
+w_{jk}^i\frac{\partial}{\partial A_{jk}^i}
+w_{jkh}^i \frac{\partial}{\partial A_{jk,h}^i},
\nonumber
\end{align}
where
\begin{align}
v_{ij}\!\!\!
& =\!\!\!-\frac{\partial u^h}{\partial x^i}y_{hj}
-\frac{\partial u^h}{\partial x^j}y_{hi},
\label{vij}\\
v_{ijk}\!\!\!
& =\!\!\!-\frac{\partial^2u^h}{\partial x^i\partial x^k}y_{hj}
-\frac{\partial^2u^h}{\partial x^j\partial x^k}y_{hi}
-\frac{\partial u^h}{\partial x^i}y_{hj,k}
-\frac{\partial u^h}{\partial x^j}y_{hi,k}
-\frac{\partial u^h}{\partial x^k}y_{ij,h}
,\label{vijk}
\end{align}
and $w_{jk}^i,w_{jkh}^i$ are given in the formulas \eqref{3_2},
\eqref{3_3}, respectively.

\subsection{$\mathrm{Diff}N$- and $\mathfrak{X}(N)$-invariance\label{invariance}}

A differential form $\omega_{r}\in \Omega^{r}(J^1(M\times _NC))$,
$r\in \mathbb{N}$, is said to be $\mathrm{Diff}N$-invariant---
or invariant under diffeomorphisms---
(resp.\ $\mathfrak{X}(N)$-invariant) if the following
equation holds: $(\bar{f}^{(1)})^\ast \omega_{r}=\omega_{r}$,
$\forall f\in \mathrm{Diff}N$ (resp.\
$L_{\bar{X}^{(1)}}\mathcal{\omega}_{r}=0$,
$\forall X\in \mathfrak{X}(N)$). Obviously,
``$\mathrm{Diff}N$-invariance'' implies
``$\mathfrak{X}(N)$-invariance'' and the converse
is almost true (see \cite{GM2}, \cite{MR}). Because
of this, below we consider $\mathfrak{X}(N)$-invariance only.

A linear frame $(X_1,\dotsc,X_N)$ at $x$ is said to
be orthonormal with respect to $g_x\in  M_x(N)$
(or simply $g_x$-orthonormal) if $g_x(X_i,X_j)=0$
for $1\leq i<j\leq n$, $g(X_i,X_i)=1$
for $1\leq i\leq n^+$, $g(X_i,X_i)=-1$
for $n^{+}+1\leq i\leq n$.

As $N$ is an oriented manifold, there exists a unique
$p$-horizontal $n$-form
$\mathbf{v}$ on $M\times _NC$ such that,
$\mathbf{v}_{(g_x,\Gamma _x)}\left( X_{1},\dotsc,X_N\right) =1$,
for every $g_x$-orthonormal basis
$(X_{1},\dotsc,X_N)$ belonging to the orientation of $N$.
Locally $\mathbf{v}=\rho v_n$, where
$\rho=\sqrt{(-1)^{n^{-}}\det(y_{ij})}$ and
$v_n=dx^1\wedge\cdots\wedge dx^n$. As proved in \cite[Proposition 7]{MV},
the form $\mathbf{v}$ is $\mathrm{Diff}N$-invariant and hence
$\mathfrak{X}(N)$-invariant. A Lagrangian density
$\Lambda $ on $J^1(M\times _NC)$ can be globally written as
$\Lambda=\mathcal{L}\mathbf{v}$ for
a unique function $\mathcal{L}\in  C^\infty (J^1(M\times _NC))$
and $\Lambda $ is $\mathfrak{X}(N)$-invariant if and only if
the function $\mathcal{L}$ is. Therefore, the invariance of
Lagrangian densities is reduced to that of scalar functions.

\begin{proposition}
\label{proposition1}A function
$\mathcal{L}\in  C^\infty (J^1(M\times _NC))$ is
$\mathfrak{X}(N)$-invariant if and only if
the following system of partial differential equations hold:
\begin{equation}
\begin{array}
[c]{ll}
0=X^i(\mathcal{L}), & \forall i,\\
0=X_h^i\left(  \mathcal{L}\right)  , & \forall h,i,\\
0=X_h^{ik}\left(  \mathcal{L}\right)  , & \forall h,i\leq k,\\
0=X_i^{jkh}\left(  \mathcal{L}\right)  , & \forall i,j\leq k\leq h,
\end{array}
\label{S}
\end{equation}
where{\small
\[
X^i=\frac{\partial}{\partial x^i},\;\forall i,
\]
\begin{align}
X_{h}^i
& =-y_{hi}\dfrac{\partial}{\partial y_{ii}}-y_{hj}
\dfrac{\partial }{\partial y_{ij}}
-y_{ih,k}\dfrac{\partial}{\partial y_{ii,k}}-y_{hj,k}
\dfrac{\partial}{\partial y_{ij,k}}
-\sum_{s\leq j}y_{sj,h}
\dfrac{\partial }{\partial y_{sj,i}}
\nonumber\\
& +A_{jk}^i\dfrac{\partial}{\partial A_{jk}^h}
-A_{jh}^r\dfrac{\partial }{\partial A_{ji}^r}
-A_{hk}^r\dfrac{\partial}{\partial A_{ik}^r}
\nonumber\\
& +A_{jk,s}^i\dfrac{\partial}{\partial A_{jk,s}^h}
-A_{jh,r}^s\dfrac{\partial}{\partial A_{ji,r}^s}
-A_{hk,r}^s\dfrac{\partial}{\partial A_{ik,r}^s}
-A_{jk,h}^r\dfrac{\partial}{\partial A_{jk,i}^r},
\;
\forall
h,i,\nonumber
\end{align}
\begin{align}
X_h^{ik}
& =-y_{ih}\dfrac{\partial}{\partial y_{ii,k}}
-y_{kh}\dfrac{\partial}{\partial y_{kk,i}}
-y_{hj}\dfrac{\partial}{\partial y_{ij,k}}
-y_{hj}\dfrac{\partial}{\partial y_{kj,i}}
-\dfrac{\partial}{\partial A_{ik}^h}
-\dfrac{\partial}{\partial A_{ki}^h}
\label{Xikh}\\
& +A_{js}^k\dfrac{\partial}{\partial A_{js,i}^h}
-A_{jh}^s\dfrac {\partial}{\partial A_{jk,i}^{s}}
-A_{hr}^s\dfrac{\partial}{\partial
A_{kr,i}^s}\nonumber\\
& +A_{js}^i\dfrac{\partial}{\partial A_{js,k}^h}
-A_{jh}^s\dfrac{\partial}{\partial A_{ji,k}^{s}}
-A_{hr}^s\dfrac{\partial}{\partial A_{ir,k}^{s}},
\;
\forall h,i\leq k,
\nonumber
\end{align}
\begin{equation}
X_{i}^{jkh}=\frac{\partial}{\partial A_{jk,h}^i}
+\frac{\partial}{\partial A_{jh,k}^i}
+\frac{\partial}{\partial A_{hk,j}^i}
+\frac{\partial}{\partial A_{hj,k}^i}
+\frac{\partial}{\partial A_{kj,h}^i}
+\frac{\partial}{\partial A_{kh,j}^i},
\;
\forall i,h\leq j\leq k.
\label{Xijkh}
\end{equation}
}
Moreover, the vector fields $X^i,X_h^i,X_h^{ik},X_i^{jkh}$ are
linearly independent and they span an involutive distribution on
$J^1(M\times _NC)$ of rank $n\tbinom{n+3}{3}$. Hence, the number of
functionally invariant Lagrangians on $J^1(M\times _NC)$ is
\[
\tfrac{1}{6}
\left(
 5n^{4}+3n^3-5n^2+3n
\right) .
\]

\end{proposition}

\begin{proof}
According to the formula \eqref{X_bar^1},
$\mathcal{L}$ is invariant if and only if,
\[
\begin{array}
[c]{l}
u^i\dfrac{\partial\mathcal{L}}{\partial x^i}
+\sum_{i\leq j}v_{ij}
\dfrac{\partial\mathcal{L}}{\partial y_{ij}}
+\sum_{i\leq j}v_{ijk}
\dfrac{\partial\mathcal{L}}{\partial y_{ij,k}}
+w_{jk}^i\dfrac{\partial \mathcal{L}}{\partial A_{jk}^i}
+w_{jkh}^i
\dfrac{\partial\mathcal{L}}{\partial A_{jk,h}^i}
=0,
\smallskip\\
\forall u^i\in  C^\infty (N),
\end{array}
\]
and expanding on this equation by using the formulas
\eqref{vij}, \eqref{vijk}, \eqref{3_2}, and \eqref{3_3}
we obtain
\begin{align*}
0  & =u^i\dfrac{\partial\mathcal{L}}{\partial x^i}\\
& +\frac{\partial u^h}{\partial x^i}
\left(
-y_{hi}\dfrac{\partial \mathcal{L}}{\partial y_{ii}}
-y_{hj}\dfrac{\partial\mathcal{L}}{\partial y_{ij}}
-y_{ih,k}\dfrac{\partial\mathcal{L}}{\partial y_{ii,k}}
-y_{hj,k}\dfrac{\partial\mathcal{L}}{\partial y_{ij,k}}
\right. \\
& -\sum_{s\leq j}y_{sj,h}
\dfrac{\partial\mathcal{L}}{\partial y_{sj,i}}
+A_{jk}^i\dfrac{\partial\mathcal{L}}{\partial A_{jk}^h}
-A_{jh}^r\dfrac{\partial\mathcal{L}}{\partial A_{ji}^r}
-A_{hk}^r\dfrac{\partial\mathcal{L}}{\partial A_{ik}^r}\\
& \left.
+A_{jk,s}^i
\dfrac{\partial\mathcal{L}}{\partial A_{jk,s}^h}
-A_{jh,r}^s\dfrac{\partial\mathcal{L}}{\partial A_{ji,r}^s}
-A_{hk,r}^s\dfrac{\partial\mathcal{L}}{\partial A_{ik,r}^s}
-A_{jk,h}^r\dfrac{\partial\mathcal{L}}{\partial A_{jk,i}^r}
\right) \\
& +\frac{\partial^2u^h}{\partial x^i\partial x^k}
\left(
-y_{ih}
\dfrac{\partial\mathcal{L}}{\partial y_{ii,k}}-y_{hj}
\dfrac{\partial \mathcal{L}}{\partial y_{ij,k}}
-\dfrac{\partial\mathcal{L}}{\partial A_{ik}^h}
\right. \\
& \left.
+A_{js}^k
\dfrac{\partial\mathcal{L}}{\partial A_{js,i}^h}
-A_{jh}^s\dfrac{\partial\mathcal{L}}{\partial A_{jk,i}^s}
-A_{hr}^r\dfrac{\partial\mathcal{L}}{\partial A_{kr,i}^r}
\right) \\
& -\frac{\partial^3u^i}{\partial x^h\partial x^k\partial x^j}
\dfrac{\partial\mathcal{L}}{\partial A_{jk,h}^i}.
\end{align*}
This equation is equivalent to the system of the statement
as the values for $u^h$, $\partial u^h/\partial x^i$,
$\partial^2u^h/\partial x^i\partial x^j$ ($i\leq j$), and
$\partial^3u^h/\partial x^i\partial x^j\partial x^k$
($i\leq j\leq k$) at a point $x\in  N$ can be taken arbitrarily.
Moreover, assume a linear combination holds
\begin{equation}
\begin{array}
[c]{l}
\lambda_a X^a
+\lambda_b^a X_a^b
+\sum_{b\leq c}\lambda_{bc}^aX_a^{bc}
+\sum_{b\leq c\leq d}\lambda_{bcd}^aX_a^{bcd}
=0,
\medskip\\
\lambda_a,\lambda_b^a,\lambda_{bc}^a,
\lambda_{bcd}^a\in  C^\infty (J^1(M\times _NC)).
\end{array}
\label{LinearCombination}
\end{equation}
By applying \eqref{LinearCombination} to $x^a$
(resp.\ $y_{ab}$) we obtain
$\lambda_a=0$ (resp.\ $\lambda_b^a=0$); again by applying
\eqref{LinearCombination} to $A_{bc}^a$, $b\leq c$
(resp.\ $A_{bc}^a$, $c\leq b $) and taking the expressions
of the vector fields \eqref{Xikh} and \eqref{Xijkh} into account,
we obtain $\lambda_{bc}^a=0$, $b\leq c$ (resp.\
$\lambda_{bc}^a=0$, $c\leq b$). Hence, \eqref{LinearCombination}
reads $\sum_{b\leq c\leq d}\lambda_{bcd}^aX_a^{bcd}=0$,
and by applying it to $A_{bc,d}^a$ and taking
the expressions of the vector fields \eqref{Xijkh}
into account, we finally obtain $\lambda_{bcd}^a=0$.
The distribution
\[
\mathcal{D}_{M\times _NC}
=\left\{
\bar{X}_{
\left(
j_x^1g,j_x^1s_{\Gamma}
\right)
}^{(1)}:
X\in \mathfrak{X}(N),
\left(
j_x^1g,j_x^1s_{\Gamma}
\right)
\in  J^1(M\times _NC)
\right\}
\]
in $T\left(  J^1(M\times _NC)\right) $, where $\bar{X}^{(1)}$
is defined in \eqref{X_bar^1}, is involutive as
\[
\left[
\bar{X}^{(1)},\bar{Y}^{(1)}
\right]
=\overline{
\left[  X,Y\right] }^{(1)},
\quad\forall X,Y\in \mathfrak{X}(N),
\]
and it is spanned by $X^i,X_h^i,X_h^{ik},X_i^{jkh}$, as proved by
the formulas above. The rest of the statement follows from the following
identities:
\begin{multline*}
\#\left\{  X^i;X_h^i;X_h^{ik},i\leq k;X_i^{jkh},h\leq j\leq
k:h,i,j,k=1,\dotsc,n\right\} \\
=n+n^2+n\tbinom{n+1}{2}+n\tbinom{n+2}{3}=n\tbinom{n+3}{3},
\end{multline*}
\[
\dim J^1
\left(
M\times _NC
\right)
-n\tbinom{n+3}{3}=\tfrac{1}{6}
\left(
5n^4+3n^3-5n^2+3n
\right) .
\]

\end{proof}

\section{Invariance of covariant Hamiltonians}

\subsection{Position of the problem}

On the bundle $E=M\times _NC$, an Ehresmann connection
can locally be written as follows:
\begin{equation}
\begin{array}
[c]{l}
\gamma=\sum_{i\leq j}
\left(
dy_{ij}+\gamma _{ijk}dx^k
\right)
\otimes
\dfrac{\partial}{\partial y_{ij}}
+\left(
dA_{jk}^i+\gamma _{jkl}^idx^l
\right)
\otimes \dfrac{\partial}{\partial A_{jk}^i},\\
\gamma _{ijk},\gamma _{jkl}^i\in  C^\infty (M\times _NC).
\end{array}
\label{Con_MxC}
\end{equation}
In particular, for a Lagrangian density $\Lambda $
on $J^1(M\times _NC)$ we obtain
\[
\Lambda^\gamma
=\left(
\sum_{i\leq j}\Bigl(\gamma _{ijk}+y_{ij,k}\Bigr)
\frac{\partial L}{\partial y_{ij,k}}+\Bigl(\gamma _{jkl}^i
+A_{jk,l}^i\Bigr)\frac{\partial L}{\partial A_{jk,l}^i}-L
\right)
dx^1 \wedge\cdots\wedge dx^n,
\]
or equivalently, $\mathcal{L}^\gamma =D^\gamma (\mathcal{L})-\mathcal{L}$,
where
\[
D^\gamma =\sum_{i\leq j}\Bigl(\gamma _{ijk}+y_{ij,k}\Bigr)
\frac{\partial }{\partial y_{ij,k}}
+\Bigl(
\gamma _{jkl}^i+A_{jk,l}^i
\Bigr)
\frac{\partial }{\partial A_{jk,l}^i}.
\]
\begin{remark}
\label{remark1}The horizontal form
$(p_0^1)^\ast \gamma-\theta=
\left(
\gamma _i^\alpha +y_i^\alpha
\right)
dx^i\otimes\partial/\partial y^\alpha $
can also be viewed as the $p_0^1$-vertical vector field
\begin{equation}
D^\gamma =
\left(
\gamma _i^\alpha +y_i^\alpha
\right)
\frac{\partial }{\partial y_i^\alpha },
\label{Dgamma}
\end{equation}
taking the natural isomorphism
$V(p_0^1)\cong(p_0^1)^\ast (p^\ast T^\ast N\otimes V(p))$
into account (cf.\ \cite{MSarda}, \cite{MSh},
\cite{Sarda}, \cite{Saunders}).
\end{remark}

According to the previous formulas, this means:
If the system \eqref{S} holds for a Lagrangian function
$\mathcal{L}$, then it also holds for the covariant
Hamiltonian $\mathcal{L}^\gamma $.

If $X\in \{X^i,X_h^i,X_h^{ik},X_i^{jkh}\}$, then
$X\left(
\mathcal{L}^\gamma
\right)  =X
\left(
D^\gamma (\mathcal{L})
\right) $, as $\mathcal{L}$ is assumed to be invariant
and hence $X(\mathcal{L})=0$. Therefore
\begin{align*}
X\left(
\mathcal{L}^\gamma
\right)   & =X
\left(
D^\gamma (\mathcal{L}
)\right) \\
& =
\left[  X,D^\gamma \right]
(\mathcal{L}),
\end{align*}
and we conclude the following:

\begin{proposition}
\label{proposition2}
The property \emph{\textbf{(P)}} holds for an Ehresmann
connection $\gamma$ on $M\times _NC$ if and only
if the vector field $D^\gamma $ transforms the sections
of the distribution $\mathcal{D}_{M\times _NC}$
into themselves, namely,
$[D^\gamma ,\Gamma(\mathcal{D}_{M\times _NC})]
\subseteq\Gamma(\mathcal{D}_{M\times _NC})$.
\end{proposition}

The problem thus reduces to compute the brackets
$\left[  X^i,D^\gamma \right]  $,
$\left[  X_h^i,D^\gamma \right]  $,
$\left[  X_h^{ik},D^\gamma \right]  $,
and $[X_i^{jkh},D^\gamma ]$. We have
\begin{align}
\left[  X^h,D^\gamma \right]
& =\sum_{i\leq j}
\frac{\partial\gamma _{ijk}}{\partial x^h}
\frac{\partial}{\partial y_{ij,k}}
+\frac{\partial\gamma _{jkl}^i}{\partial x^h}
\frac{\partial}{\partial A_{jk,l}^i},
\label{XhDgamma}\\
\left[
X_b^{cda},D^\gamma
\right]
& =X_b^{cda},\quad\forall b,c\leq d\leq a,
\nonumber
\end{align}
\begin{align}
\quad
\left[
X_h^i,D^\gamma
\right]
& =\sum_{a\leq b}Y_h^i
\left(
\gamma _{abk}
\right)
\frac{\partial}{\partial y_{ab,k}}
+\sum_{i\leq h}\gamma _{ihk}\dfrac{\partial}{\partial y_{ii,k}}
+\sum_{h<i}\gamma _{hik}\dfrac{\partial}{\partial y_{ii,k}}
\label{XhiDgamma1}\\
& +\sum_{h\leq j}\gamma _{hjk}
\dfrac{\partial}{\partial y_{ij,k}}
+\sum _{j<h}\gamma _{jhk}\dfrac{\partial}{\partial y_{ij,k}}
+\sum_{a\leq b} \gamma _{abh}\dfrac{\partial}{\partial y_{ab,i}}
\nonumber\\
& +\left(
Y_h^i
\left(
\gamma _{bcr}^a
\right)
-\delta_a^h\gamma _{bcr}^i
+\delta_i^c\gamma _{bhr}^a
+\delta_i^b\gamma _{hcr}^a
+\delta_i^r\gamma _{bch}^a
\right)
\dfrac{\partial}{\partial A_{bc,r}^a},
\nonumber
\end{align}
\begin{equation}
\left[
X_h^{ik},D^\gamma
\right]
=\sum_{a\leq b}Y_h^{ik}
\left(
\gamma _{abc}
\right)
\frac{\partial}{\partial y_{ab,c}}+Y_h^{ik}
\left(
\gamma _{abc}^d
\right)
\frac{\partial}{\partial A_{ab,c}^d}
+X_h^{ik}-Y_h^{ik},
\label{bracket3}
\end{equation}
where
\[
Y_h^i=-y_{hi}
\dfrac{\partial}{\partial y_{ii}}
-y_{hj}
\dfrac{\partial }{\partial y_{ij}}
+A_{jk}^i
\dfrac{\partial}{\partial A_{jk}^h}
-A_{jh}^r
\dfrac{\partial}{\partial A_{ji}^r}
-A_{hk}^r
\dfrac{\partial}{\partial A_{ik}^r},
\]
\[
Y_h^{ik}=-\frac{\partial}{\partial A_{ik}^h}
-\frac{\partial}{\partial A_{ki}^h},
\]
and the following formula has been used:
\[
\dfrac{\partial y_{rs,k}}{\partial y_{ij,h}}=\delta_h^k
\left(
\delta _i^r\delta_j^s
+\delta_j^r\delta_i^s
-\delta_j^i\delta _r^i\delta_s^j
\right) .
\]

\subsection{The class of the Ehresmann connections defined}

Let $p\colon M\times _NC\to N$,
$\mathrm{pr}_1\colon M\times _NC\to M$,
$\mathrm{pr}_2\colon M\times _NC\to C$ be the
natural projections. By taking the differential of
$\mathrm{pr}_1$ and $\mathrm{pr}_2$, a natural
identification is obtained
$T(M\times _NC)=TM\times _{TN}TC$. Hence
\begin{align*}
V(p)
& =V(p_{M})\times _NV(p_C)\\
& =\mathrm{pr}_1^\ast V(p_{M})
\oplus \mathrm{pr}_2^\ast V(p_C)
\end{align*}
and two unique vector-bundle homomorphisms exist
\[
\gamma _M\colon \mathrm{pr}_1^\ast TM
\to \mathrm{pr}_1^\ast V(p_M),
\quad
\gamma _C\colon \mathrm{pr}_2^\ast TC
\to \mathrm{pr}_2^\ast V(p_C),
\]
such that,
\[
\begin{array}
[c]{ll}
\gamma(X)=
\left(
\gamma _M
\left(
\mathrm{pr}_{1\ast}X
\right)
,\gamma _C
\left(
\mathrm{pr}_{2\ast}X
\right)
\right)
, & \forall X\in  T(M\times _NC),\\
\gamma _M(Y)=Y, &
\forall Y\in \mathrm{pr}_1^\ast V(p_M),\\
\gamma _C(Z)=Z, & \forall Z\in \mathrm{pr}_2^\ast V(p_C).
\end{array}
\]
If $\gamma$ is given by the local expression
of the formula \eqref{Con_MxC}, then
\[
\begin{array}
[c]{l}
\gamma _{M}=\sum_{i\leq j}
\left(
dy_{ij}+\gamma _{ijk}dx^k
\right)
\otimes\dfrac{\partial}{\partial y_{ij}},\;\gamma _C
=\left(
dA_{jk}^i+\gamma _{jkl}^idx^l
\right)
\otimes\dfrac{\partial}{\partial A_{jk}^i},\medskip\\
\gamma _{ijk},\gamma _{jkl}^i\in  C^\infty (M\times _NC).
\end{array}
\]

\subsubsection{The first geometric condition on $\gamma$
\label{first_condition}}

Let $q\colon F(N)\to M$ be the projection given by
\begin{align}
q(X_1,\dotsc,X_N)  & =g_x\label{map_q}\\
& =\varepsilon_hw^h\otimes w^h,\nonumber
\end{align}
where $(w^1,\dotsc,w^n)$ is the dual coframe of
$(X_1,\dotsc,X_N)\in  F_x(N)$, i.e., $g_x$ is the metric
for which $(X_1,\dotsc,X_N)$ is a $g_x$-orthonormal
basis and $\varepsilon_h=1$ for $1\leq h\leq n^{+}$,
$\varepsilon_h=-1$ for $n^{+}+1\leq h\leq n$. As readily seen,
$q$ is a principal $G$-bundle with $G=O(n^{+},n^{-})$.

Given a linear connection $\Gamma$ and a tangent vector
$X\in  T_xN$, for every $u$ in $p^{-1}(x)$ there exists
a unique $\Gamma$-horizontal tangent vector
$X_u^{h_\Gamma}\in  T_u(FN)$ such that,
$(p_F)_\ast X_u^{h_\Gamma }=X$. The local expression
for the horizontal lift is known
to be (\cite[Chapter III, Proposition 7.4]{KN}),
\begin{equation}
\left(
\frac{\partial}{\partial x^j}
\right) ^{h_\Gamma}
=\frac{\partial }{\partial x^j}
-\Gamma _{jk}^ix_l^k
\frac{\partial}{\partial x_l^i}.
\label{horizontal_lift}
\end{equation}

\begin{lemma}
\label{lemma1}Given a metric $g_x\in  p_M^{-1}(x)$, let
$u\in  p_F^{-1}(x)$ be a linear frame such that $q(u)=g_x$.
The projection $q_\ast (X_u^{h_{\Gamma _x}})$ does not
depend on the linear frame $u$ chosen over $g_x$.
\end{lemma}

\begin{proof}
In fact, any other linear frame projecting onto $g_x$
can be written as $u\cdot A$, $A\in  G$. As the horizontal
distribution is invariant under right translations
(see \cite[II, Proposition 1.2]{KN}), the following equation
holds:
$\left(
R_A
\right) _\ast
\left(
X_u^{h_\Gamma }
\right)
=X_{u\cdot A}^{h_\Gamma}$. Hence
\begin{align*}
q_\ast
\left(
X_{u\cdot A}^{h_\Gamma}
\right)   & =q_\ast
\left(
\left(
R_A
\right)  _\ast
\left(
X_u^{h_\Gamma}
\right)
\right) \\
& =\left(
q\circ R_A
\right)  _\ast
\left(  X_u^{h_\Gamma}
\right) \\
& =q_\ast
\left(
X_u^{h_\Gamma}
\right)  .
\end{align*}

\end{proof}

\begin{proposition}
\label{proposition0}An Ehresmann connection $\gamma$ on
$M\times _NC$\ satisfies the following condition:

\begin{description}
\item[$(C_M)$] $\gamma _M
\left(
\left(  g_x,\Gamma _x
\right)
,X
\right)
=X-q_\ast
\left(
\left(
(p_M)_\ast(X)
\right) _u^{h_{\Gamma _x}}
\right) ,
\quad $
\end{description}
\noindent
$\forall X\in  T_{g_x}M,\;u\in  q^{-1}(g_x)$,
(which does not depend on the linear frame
$u\in  q^{-1}(g_x)$ chosen, according to
\emph{Lemma \ref{lemma1}}) if and only if the following
equations hold:
\begin{equation}
\gamma _{klj}=-
\left(
y_{al}A_{jk}^a+y_{ak}A_{jl}^a
\right) ,
\label{gamma_klj}
\end{equation}
where the functions $\gamma _{klj}$ (resp.\ $y_{ij}$, resp.\ $A_{jk}^i$)
are defined in the formula \emph{\eqref{Con_MxC}}
(resp.\ \emph{\ \eqref{inducedsystemM(N)},} resp.\
\emph{\eqref{inducedsystemC(N)}}).
\end{proposition}

\begin{proof}
Letting
$(\chi_j^i)_{i,j=1}^n= \left( (x_j^i)_{i,j=1}^n \right) ^{-1}$,
the dual coframe of the linear frame $u=(X_1,\dotsc,X_N)\in  F_x(N)$
given in \eqref{frame} is $(w^1,\dotsc,w^n)$,
$w^h=\chi _k^h(u)\left(  dx^k\right)  _x$, $1\leq h\leq n$,
and the projection $q$ is given by
\begin{align*}
q(u)  & =g_x\\
& =\sum\nolimits_{h=1}^n\varepsilon_{h}\chi_k^h(u)\chi_l^h(u)
\left(
dx^k
\right)  _x\otimes
\left(
dx^l
\right)  _x.
\end{align*}
Therefore the equations of the projection
\eqref{map_q} are as follows:
\[
\begin{array}
[c]{l}
x^i\circ q=x^i,\\
y_{kl}\circ q=\sum\nolimits_{h=1}^n\varepsilon_{h}\chi_{k}^h\chi_{l}^h.
\end{array}
\]
Hence
\[
q_\ast
\left(
\frac{\partial}{\partial x_b^a}
\right)  _{u}
=\sum_{k\leq l}\varepsilon_h
\left\{
\frac{\partial\chi_k^h}{\partial x_b^a}\chi_l^h
+\chi_k^h\frac{\partial\chi_l^h}{\partial x_b^a}
\right\} (u)
\left(
\frac{\partial}{\partial y_{kl}}
\right) _{g_x}.
\]
Taking derivatives with respect to $x_b^a$ on the identity
$\chi_r^hx_i^{r}=\delta_i^h$, multiplying the outcome
by $\chi_k^i $, and summing up over the index $i$,
the following formula is obtained:
$\partial\chi_k^h/\partial x_b^a=-\chi_a^h\chi_k^b$.
Replacing this equation into the expression for
$q_\ast \left( \partial/\partial x_b^a \right) _u$
above, we have
\[
q_\ast
\left(
\frac{\partial}{\partial x_b^a}
\right) _u
=-\sum_{k\leq l}
\left\{
\chi_k^b(u)y_{al}
\left(
g_x
\right)
+\chi_l^b(u)y_{ak}
\left(
g_x\right)
\right\}
\left(
\frac{\partial}{\partial y_{kl}}
\right)  _{g_x}.
\]
From \eqref{horizontal_lift}, evaluated at $u\in  q^{-1}(g_x)$, we deduce
\[
\begin{array}
[c]{ll}
q_{\ast}
\left(
\dfrac{\partial}{\partial x^j}
\right)  _u^{h_\Gamma} &
\!\!\!=
\left(  \!
\dfrac{\partial}{\partial x^j}\!
\right)  _{g_x }\!\!\!
-\Gamma _{jc}^a(x)x_b^c(u)q_\ast
\left(  \!
\dfrac{\partial }{\partial x_b^a}\!
\right)  _{g_x}\medskip\\
& \!\!\!=
\left(  \!\dfrac{\partial}{\partial x^j}\!
\right) _{g_x}
\medskip\\
& +\sum_{k\leq l}\Gamma _{jc}^a(x)x_b^c(u)
\left\{
\chi_k^b (u)y_{al}
\left(  g_x
\right)
+\chi_l^b(u)y_{ak}
\left(  g_x
\right)
\right\}
\left(  \!
\dfrac{\partial}{\partial y_{kl}}\!
\right)  _{g_x}
\medskip\\
& \!\!\!=
\left(  \!\dfrac{\partial}{\partial x^j}\!
\right)  _{g_x}\!\!\!+\sum_{k\leq l}
\left\{
\Gamma _{jk}^a(x)y_{al}
\left(
g_x
\right)
+\Gamma _{jl}^a(x)y_{ak}
\left(
g_x
\right)
\right\}  \!
\left(
\!\dfrac{\partial}{\partial y_{kl}}\!
\right) _{g_x}.
\end{array}
\]
The condition $(C_M)$ holds automatically whenever $X\in  V(p_M)$.
Hence, $(C_M)$ holds if and only if it holds for
$X=(\partial/\partial x^j)_{g_x}$, namely,
\begin{align*}
\sum_{k\leq l}\gamma _{klj}(g_x,\Gamma _x)
\left(
\dfrac{\partial}{\partial y_{kl}}
\right)  _{g_x} &
=\gamma _M
\left(
\left(
g_x,\Gamma _x
\right)  ,
\left(
\frac{\partial}{\partial x^j}
\right)  _{g_x}
\right)
\\
&  =\left(
\frac{\partial}{\partial x^j}
\right)  _{g_x}
-q_\ast
\left(
\frac{\partial}{\partial x^j}
\right)  _u^{h_{\Gamma _x}}\\
&  =-\sum_{k\leq l}
\left\{
\Gamma _{jk}^a(x)y_{al}
\left(  g_x
\right)
+\Gamma _{jl}^a(x)y_{ak}
\left(  g_x
\right)
\right\}  \!
\left(
\frac{\partial}{\partial y_{kl}}
\right)  _{g_x}\!,
\end{align*}
thus proving the formula \eqref{gamma_klj} in the statement.
\end{proof}

\subsubsection{The canonical covariant derivative\label{canonical_cov_deriv}}

As is known (e.g., see \cite[III, section 1]{KN},
\cite[pp. 157--158]{MSarda}) every connection $\Gamma$
on a principal $G$-bundle $P\to N$ induces a covariant derivative
$\nabla^\Gamma $ on the vector bundle associated to $P$
under a linear representation $\rho\colon G\to Gl(m,\mathbb{R})$
with standard fibre $\mathbb{R}^m$. In particular,
this applies to the principal bundle of linear frames,
thus proving that every linear connection $\Gamma$ on $N$
induces a covariant derivative $\nabla^\Gamma $ on every tensorial
vector bundle $E\to N$.

The bundles $(p_C)^\ast E$, where $E$ is
a tensorial vector bundle, are endowed with a canonical covariant
derivative $\nabla^E$ completely determined by the formula:
\begin{equation}
\left(
\left(
\nabla^E
\right)  _X(f\xi)
\right)  (\Gamma _x)
=\left(
\left(  Xf
\right)  \xi
\right)
\left(
\Gamma _x
\right)
+f
\left(
\Gamma _x
\right)
\left(
\nabla_{(p_C)_\ast X}^{\Gamma _x}\xi
\right) (x),
\label{canonical}
\end{equation}
for all $X\in  T_{\Gamma _x}C$, $f\in  C^\infty (C)$,
and every local section $\xi$ of $E$ defined on a
neighbourhood of $x$. The uniqueness of $\nabla^E$
follows from \eqref{canonical} as the sections of
$E$ span the sections of $(p_C)^\ast E$ over
$C^\infty (C)$, see \cite[0.3.6]{DPV}. Below, we are
specially concerned with the cases $E=TN$ and
$E=\wedge^2T^\ast N\otimes TN$.

\subsubsection{The $2$-form associated with $\gamma _C$}

As $p_C\colon C\to N$ is an affine bundle modelled over
$\otimes ^2T^\ast N\otimes TN$, there is a natural identification
\[
V(p_C)\cong
\left(  p_C
\right)  ^\ast
\left(
\otimes^2T^\ast N\otimes TN
\right)
\]
and consequently, an Ehresmann connection
$\gamma _C$ on $C$ can also be viewed as a homomorphism
$\gamma _C\colon TC\to\otimes^2T^{\ast }N\otimes TN$.
If $\gamma _C$ is locally given by
\begin{equation}
\gamma _C=
\left(
dA_{jk}^i+\gamma _{jkl}^idx^l
\right)  \otimes
\dfrac{\partial}{\partial A_{jk}^i},
\quad\gamma _{jkl}^i\in  C^{\infty }(C),
\label{gamma_C}
\end{equation}
then
\[
\gamma _C=(dA_{jk}^i
+\gamma _{jkl}^idx^l)\otimes dx^j\otimes dx^k\otimes
\frac{\partial}{\partial x^i},
\]
and $\gamma _C$ induces a $2$-form $\tilde{\gamma}_C$ taking values in
$(p_C)^\ast (T^\ast N\otimes TN)$ as follows:
\[
\begin{array}
[c]{l}
\tilde{\gamma}_C(X,Y)=c_1^1
\left(
(p_C)_\ast(Y)\otimes\gamma _C
\left(
X
\right)
\right)
-c_1^1
\left(
(p_C)_\ast(X)\otimes \gamma _C
\left(
Y
\right)
\right)  ,\smallskip\\
\forall X,Y\in  T_{\Gamma _x}C,
\end{array}
\]
where
\[
\begin{array}
[c]{l}
c_1^1\colon TN\otimes T^\ast N\otimes T^\ast N\otimes TN
\to T^\ast N\otimes TN,\\
c_1^1
\left(
X_1\otimes w_1\otimes w_2\otimes X_2
\right)
=w_1(X_1)w_2\otimes X_2,\\
X_1,X_2\in  T_xN,\;w_1,w_2\in  T_x^\ast N.
\end{array}
\]

If $\gamma _C$ is given by \eqref{gamma_C},
then from the very definition of $\tilde{\gamma}_C$
the following local expression is obtained:
\[
\tilde{\gamma}_C=
\left(
dA_{lh}^c+
\left(
\gamma _{lha}^c-\gamma _{ahl}^c
\right)  dx^a
\right)  \wedge dx^l\otimes dx^h\otimes
\frac{\partial}{\partial x^c}.
\]

\subsubsection{The second geometric condition on $\gamma$}

Let $\mathrm{alt}_{12}\colon\otimes^2T^\ast N\otimes TN
\to \wedge^2T^\ast N\otimes TN$ be the operator alternating
the two covariant arguments.

The vector bundle
$(p_C)^\ast \left(  \wedge^2T^\ast N\otimes TN\right) $
admits a canonical section
\[
\begin{array}
[c]{l}
\tau_N\colon C\to\wedge^2T^\ast N\otimes TN,\\
\tau_N
\left(  \Gamma _x
\right)
=T^{^{\Gamma _x}},\;\forall\Gamma _x\in  C,
\end{array}
\]
where $T^{^{\Gamma _x}}$ is the torsion of $\Gamma _x$.
Locally,
\[
\tau_N=\sum_{j<k}(A_{jk}^i-A_{kj}^i)dx^j\wedge dx^k\otimes
\frac{\partial}{\partial x^i}.
\]
From the previous formulas the next result follows:

\begin{proposition}
\label{proposition3}Let $\gamma$ be an Ehresmann connection on
$M\times _NC$, let $\nabla^{(1)}=\nabla^{E_{1}}$ with
$E_{1}=TN$, let $R^{\nabla^{(1)}}$ be its curvature form,
and finally, let $\nabla^{(2)}=\nabla^{E_2}$ with
$E_2=\wedge^2T^\ast N\otimes TN$.

\begin{enumerate}
\item[$(C_C)$] Assume the component $\gamma _C$ of
$\gamma$ is defined on $C$. Then, the equations
\begin{equation}
\tilde{\gamma}_C=R^{\nabla^{(1)}},
\label{b}
\end{equation}
\begin{equation}
\mathrm{alt}_{12}\circ\gamma _C=\nabla^{(2)}\tau_N,
\label{c}
\end{equation}
are locally equivalent to the following ones:
\begin{equation}
\gamma _{str}^h-\gamma _{rts}^h
=A_{rm}^hA_{st}^m-A_{sm}^hA_{rt}^m,
\label{28}
\end{equation}
\begin{align}
\gamma _{rst}^h-\gamma _{srt}^h  & =A_{tm}^h
\left(
A_{rs}^m-A_{sr}^m
\right)
+A_{ts}^m
\left(
A_{mr}^h-A_{rm}^h
\right)
\label{29}\\
& \qquad
\qquad
\qquad
\quad
\;
+A_{tr}^m
\left(  A_{sm}^h-A_{ms}^h\right)
.\nonumber
\end{align}

\end{enumerate}
\end{proposition}

\subsection{Solution to the problem (P)}

\begin{theorem}
\label{th1}If the connection $\gamma$ on $M\times _NC$ satisfies the
conditions $(C_M)$ and $(C_C)$ introduced above, then the vector field
$D^\gamma $ satisfies the property stated in \emph{Proposition
\ref{proposition2}} and, accordingly\emph{\ }the covariant Hamiltonian with
respect to $\gamma$ of every $\mathfrak{X}(N)$-invariant Lagrangian is also
$\mathfrak{X}(N)$-invariant.
\end{theorem}

\begin{proof}
When $\gamma _M$ satisfies the condition $(C_M)$ the brackets
\eqref{XhDgamma}, \eqref{XhiDgamma1}, and \eqref{bracket3}
are respectively given by
\begin{equation}
\left[
X^h,D^\gamma
\right]
=\frac{\partial\gamma _{jkl}^i}{\partial x^h}
\frac{\partial}{\partial A_{jk,l}^i},
\label{XhDgammaC}
\end{equation}
\begin{equation}
\left[
X_h^i,D^\gamma
\right]
=\left(
Y_h^i
\left(
\gamma _{bcr}^a
\right)
-\delta _a^h\gamma _{bcr}^i+\delta_i^c
\gamma _{bhr}^a+\delta_i^b\gamma _{hcr}^a
+\delta_i^r\gamma _{bch}^a
\right)
\dfrac{\partial}{\partial A_{bc,r}^a},
\label{XhiDgamma1C}
\end{equation}
\begin{align*}
\left[
X_h^{ik},D^\gamma
\right]
& =\left(
-\frac{\partial\gamma _{abc}^d}{\partial A_{ik}^h}
+\delta_i^c
\left(
\delta_d^h
A_{ab}^k-\delta_b^kA_{ah}^d-\delta_a^kA_{hb}^d
\right)
\right.
\\ & \left.
-\frac{\partial\gamma _{abc}^d}{\partial A_{ki}^h}+\delta_k^c
\left(
\delta_d^hA_{ab}^i-\delta_b^iA_{ah}^d-\delta_a^iA_{hb}^d
\right)
\right)
\dfrac{\partial}{\partial A_{ab,c}^d}.
\end{align*}
In addition, if $\gamma _C$ satisfies the condition $(C_C)$, then taking
derivatives with respect to $x^h$ in \eqref{28} and \eqref{29} we obtain
\[
\frac{\partial\gamma _{klj}^i}{\partial x^h}
=\frac{\partial\gamma _{jlk}^i}{\partial x^h},
\quad
\frac{\partial\gamma _{jkl}^i}{\partial x^h}
=\frac{\partial\gamma _{kjl}^i}{\partial x^h},
\]
and renaming indices we deduce
\[
\begin{array}
[c]{l}
\dfrac{\partial\gamma _{jjk}^i}{\partial x^h}
=\dfrac{\partial\gamma _{jkj}^i}{\partial x^h}
=\dfrac{\partial\gamma _{kjj}^i}{\partial x^h}
\; (j<k),
\medskip\\
\dfrac{\partial\gamma _{kkj}^i}{\partial x^h}
=\dfrac{\partial\gamma _{kjk}^i}{\partial x^h}
=\dfrac{\partial\gamma _{jkk}^i}{\partial x^h}
\; (j<k),
\medskip\\
\dfrac{\partial\gamma _{jkl}^i}{\partial x^h}
=\dfrac{\partial\gamma _{klj}^i}{\partial x^h}
=\dfrac{\partial\gamma _{ljk}^i}{\partial x^h}
=\dfrac{\partial\gamma _{kjl}^i}{\partial x^h}
=\dfrac{\partial\gamma _{lkj}^i}{\partial x^h}
=\dfrac{\partial\gamma _{jlk}^i}{\partial x^h}
\; (j<k<l).
\end{array}
\]
From \eqref{XhDgammaC} we obtain
\begin{align*}
\left[
X^h,D^\gamma
\right]
& =\sum_{j<k<l}
\frac{\partial\gamma _{jkl}^i}{\partial x^h}X_i^{jkl}
+\tfrac{1}{2}\sum_{j<k}
\frac{\partial\gamma _{jjk}^i}{\partial x^h}X_i^{jjk}\\
& +\tfrac{1}{2}\sum_{j<k}
\frac{\partial\gamma _{kkj}^i}{\partial x^h} X_i^{kkj}
+\tfrac{1}{6}
\frac{\partial\gamma _{jjj}^i}{\partial x^h}X_i^{jjj},
\end{align*}
and consequently the values of $\left[  X^h,D^\gamma \right] $ belong to
the distribution $\mathcal{D}_{M\times _NC}$.

Moreover, as $\gamma _C$ is assumed to be defined on $C$, we have
\[
Y_h^i
\left(
\gamma _{bcr}^a
\right)
=\left(
\delta_h^sA_{jk}^i-\delta_k^iA_{jh}^s-\delta_j^iA_{hk}^s
\right)
\dfrac{\partial\gamma _{bcr}^a}{\partial A_{jk}^s}.
\]
For the sake of simplicity, below we set
\[
\left(
T_h^i
\right)  _{bcr}^a
\!=\!A_{jk}^i
\dfrac{\partial\gamma _{bcr}^a}{\partial A_{jk}^h}
\!-\!A_{jh}^s
 \dfrac{\partial\gamma _{bcr} ^a}{\partial A_{ji}^s}
 \!-\!A_{hk}^s
 \dfrac{\partial\gamma _{bcr}^a}{\partial A_{ik}^s}
 \!-\!\delta_a^h\gamma _{bcr}^i
 \!+\!\delta_i^b\gamma _{hcr}^a
 \!+\!\delta_i^c\gamma _{bhr}^a
 \!+\!\delta_i^r\gamma _{bch}^a.
\]
Taking derivatives with respect to $A_{jk}^s$,
the equations \eqref{28} y \eqref{29} yield
\[
\frac{\partial\gamma _{bcr}^a}{\partial A_{jk}^s}
-\frac{\partial \gamma _{rcb}^a}{\partial A_{jk}^s}
=\delta_r^j\delta_s^aA_{bc}^k
-\delta_b^j\delta_s^aA_{rc}^k
+\delta_b^j\delta_C^kA_{rs}^a
-\delta_r^j\delta_C^kA_{bs}^a,
\]
\begin{align*}
\!\!\dfrac{\partial\gamma _{rbc}^a}{\partial A_{jk}^s}
\!-\!\dfrac{\partial\gamma _{brc}^a}{\partial A_{jk}^s}
\! & =\!\delta_C^j\delta_s^aA_{rb}^k
\!-\!\delta_s^a\delta_C^jA_{br}^k
\!-\!\delta_s^a\delta_b^kA_{cr}^j
\!-\!\delta_s^a\delta_r^jA_{cb}^k
\!+\!\delta_s^a\delta_r^kA_{cb}^j
\!+\!\delta_s^a\delta_b^jA_{cr}^k\\
& \!\!+\delta_C^j\delta_b^kA_{sr}^a
\!-\!\delta_C^j\delta_r^kA_{sb}^a
\!+\!\delta_r^j\delta_b^kA_{cs}^a
\!-\!\delta_b^j\delta_r^kA_{cs}^a
\!+\!\delta_C^j\delta_r^kA_{bs}^a
\!-\!\delta_C^j\delta_b^kA_{rs}^a.
\end{align*}
From these expressions, the following symmetries of indices are obtained:
\[
\begin{array}
[c]{l}
\left(  T_h^i\right)  _{bbc}^a
=\left(  T_h^i\right)  _{bcb}^a
=\left(  T_h^i\right)  _{cbb}^a\;(b<c),
\medskip\\
\left(  T_h^i\right)  _{bcc}^a
=\left(  T_h^i\right)  _{cbc}^a
=\left(  T_h^i\right)  _{ccb}^a\;(b<c),
\medskip\\
\left(  T_h^i\right)  _{bcd}^a
=\left(  T_h^i\right)  _{dbc}^a
=\left(  T_h^i\right)  _{cdb}^a
=\left(  T_h^i\right)  _{bdc}^a
=\left(  T_h^i\right)  _{dcb}^a
=\left(  T_h^i\right)  _{cbd}^a
\; (b<c<d),
\end{array}
\]
and from \eqref{XhiDgamma1C} we obtain
\begin{align*}
\left[  X_h^i,D^\gamma \right]
& =\sum_{b<c<d}
\left(
T_h^i
\right) _{bcd}^aX_a^{bcd}
+\tfrac{1}{2}\sum_{b<c}
\left(
T_h^i
\right) _{bbc}^aX_a^{bbc}\\
& +\tfrac{1}{2}\sum_{b<c}
\left(
T_h^i
\right)  _{ccb}^aX_a^{ccb}
+\tfrac{1}{6}
\left(
T_h^i
\right)  _{bbb}^aX_a^{bbb}.
\end{align*}
Hence $\left[  X_h^i,D^\gamma \right]  $ also takes values into the
distribution $\mathcal{D}_{M\times _NC}$.

The proof for the third bracket is similar to the previous
two cases but longer. Letting
\begin{align*}
\left(
T_h^{ik}
\right) _{rbc}^a
& =-\frac{\partial\gamma _{rbc}^a}{\partial A_{ik}^h}
-\frac{\partial\gamma _{rbc}^a}{\partial A_{ki}^h}
+\delta_i^c
\left(
\delta_a^hA_{rb}^k
-\delta_b^kA_{rh}^a
-\delta_r^kA_{hb}^a
\right) \\
& \qquad\qquad\qquad\qquad+\delta_k^c
\left(
\delta_a^hA_{rb}^i-\delta_b^iA_{rh}^a
-\delta_r^iA_{hb}^a
\right)  ,
\end{align*}
the following symmetries are obtained:
\[
\begin{array}
[c]{l}
\left(
T_h^{ik}
\right) _{bbc}^a
=\left(
T_h^{ik}
\right) _{bcb}^a
=\left(
T_h^{ik}
\right) _{cbb}^a
(b<c),
\medskip\\
\left(
T_h^{ik}\right) _{bcc}^a
=\left(
T_h^{ik}\right) _{cbc}^a
=\left(
T_h^{ik}
\right) _{ccb}^a
(b<c),
\medskip\\
\left(
T_h^{ik}
\right)  _{bcd}^a
=\left(
T_h^{ik}
\right) _{dbc}^a
=\left(
T_h^{ik}
\right) _{cdb}^a
=\left(
T_h^{ik}
\right)
_{bdc}^a
=\left(
T_h^{ik}
\right) _{dcb}^a
=\left(
T_h^{ik}
\right) _{cbd}^a
(b<c<d).
\end{array}
\]
Hence
\begin{align*}
\left[  X_h^{ik},D^\gamma
\right]
& =\sum_{b<c<d}
\left(
T_h^{ik}
\right)  _{bcd}^aX_a^{bcd}
+\tfrac{1}{2}\sum_{b<c}
\left(
T_{h}^{ik}
\right) _{bbc}^aX_a^{bbc}\\
& +\tfrac{1}{2}\sum_{b<c}
\left(  T_h^{ik}
\right)  _{ccb}^aX_a^{ccb}
+\tfrac{1}{6}
\left(
T_h^{ik}
\right) _{bbb}^aX_a^{bbb},
\end{align*}
and the proof is complete.
\end{proof}

\begin{theorem}
\label{th2}The Ehresmann connections on $C$ satisfying
the equations \emph{\eqref{b}} and \emph{\eqref{c}}
are the sections of an affine bundle over $C$ modelled
over the vector bundle
$\left( p_C\right) ^\ast \left( S^{3}T^\ast N\otimes TN\right) $.
Consequently, there always exist Ehresmann
connections on $M\times _NC$ fulfilling the conditions
$(C_M)$ and $(C_C)$ introduced above.
\end{theorem}

\begin{proof}
If two Ehresmann connections $\gamma _C,\gamma _C^\prime $
satisfy the equations \eqref{b} and \eqref{c},
then the difference tensor field
$t=\gamma _C^\prime -\gamma _C$, which is a section of the bundle
$\left( p_C\right) ^\ast \left( \otimes^3T^\ast N\otimes TN\right) $,
satisfies the following symmetries:
\begin{align}
t(X_1,X_2,X_3)  & =t(X_3,X_2,X_1),\label{1-3}\\
t(X_1,X_2,X_3)  & =t(X_2,X_1,X_3),\label{1-2}
\end{align}
according to \eqref{28}, \eqref{29}, respectively, for all
$X_1,X_2,X_3\in  T_xN$, $\Gamma _x\in  C_x(N)$. Hence
\[
t(X_1,X_3,X_2)
\overset{\mathrm{\eqref{1-3}}}{=}t(X_2,X_3,X_1)
\overset{\mathrm{\eqref{1-2}}}{=}t(X_3,X_2,X_1)
\overset{\mathrm{\eqref{1-3}}}{=}t(X_1,X_2,X_3),
\]
thus proving that $t$ is totally symmetric. The second part
of the statement thus follows from the fact that an affine
bundle always admits global sections, e.g., see
\cite[I, Theorem 5.7]{KN}.
\end{proof}

\begin{remark}
The results obtained above also hold if the bundle
of linear connections is replaced by the subbundle
$C^{\mathrm{sym}}=C^{\mathrm{sym}}(N)\subset C$ of
symmetric linear connections; the only difference
to be observed between both bundles is that in the
symmetric cases the equation \eqref{c}, or equivalently
\eqref{29}, holds automatically.
\end{remark}

\section{The second-order formalism}

In this section we consider the problem
of invariance of covariant Hamiltonians
for second-order Lagrangians defined on
the bundle of metrics, i.e., for functions
$\mathcal{L}\in  C^\infty (J^2M)$, where $M$ denotes, as
throughout this paper, the bundle of pseudo-Riemannian
metrics of a given signature $(n^{+},n^{-})$ on $N$.

\subsection{Second-order Ehresmann connections\label{ss2.5}}

A second-order Ehresmann connection on
$p\colon E\to N$ is a differential $1$-form $\gamma^2$ on
$J^1E$ taking values in the vertical
sub-bundle $V(p^1)$ such that $\gamma^2(X)=X$ for every
$X\in  V(p^1)$. (We refer the reader to \cite{MR3} for
the basics on Ehresmann connections of arbitrary order.)
Once a connection $\gamma^2$ is given, we have a
decomposition of vector bundles
$T(J^1E)=V(p^1)\oplus\ker\gamma^2$,
where $\ker\gamma^2$ is called the horizontal
sub-bundle determined by $\gamma^2$. In the coordinate
system on $J^1E$ induced from a fibred coordinate
system $(x^j,y^\alpha )$ for $p$, a connection form
can be written as
\begin{equation}
\gamma^2
=(dy^\alpha +\gamma _j^\alpha dx^j)\otimes
\frac{\partial}{\partial y^\alpha }
+(dy_i^\alpha
+\gamma _{ij}^\alpha dx^j)\otimes
\frac{\partial}{\partial y_i^\alpha },
\quad\gamma _j^\alpha ,\gamma _{ij}^\alpha \in  C^\infty (J^1E).
\label{f2}
\end{equation}
As in the first-order case, the action of the group
$\mathrm{Aut}(p)$ on the space of second-order connections
is defined by the formula
\[
\Phi\cdot\gamma^2=
\left(  \Phi ^{(1)}
\right)  _\ast
\circ \gamma ^2\circ
\left(  \Phi ^{(1)}
\right)  _\ast^{-1},
\quad
\forall \Phi \in  \mathrm{Aut}(p).
\]
As $\Phi^{(1)}\colon J^1M\to J^1M$ is a morphism
of fibred manifolds over $N$,
$(\Phi^{(1)})_\ast$ transforms the vertical
subbundle $V(p^1)$ into itself;
hence the previous definition makes sense.

\subsection{A remarkable isomorphism}

\begin{theorem}
\label{one-to-one}Let $\Gamma^g$ be the
Levi-Civita connection of a pseudo-Riemannian metric
$g$ on $N$. The mapping
$\zeta_N\colon J^1M\to M\times _NC^{\mathrm{sym}}$,
$\zeta_N(j_x^1g)=(g_x,\Gamma _x^g)$ is a diffeomorphism.
There is a natural one-to-one
correspondence between first-order Ehresmann connections
on the bundle $p\colon M\times _NC^{\mathrm{sym}}\to N$
and second-order Ehresmann connections on the bundle
$p_M\colon M\to N$, which is explicitly
given by,
\begin{equation}
\gamma^2=
\left(
\left(
\zeta_N^v
\right)
_\ast
\right) ^{-1}
\circ\gamma\circ
\left(
\zeta_N
\right)  _\ast,
\label{gamma2}
\end{equation}
where
$\gamma\colon T(M\times _NC^{\mathrm{sym}})\to V(p)$
is a first-order Ehresmann connection,
\[
\left(
\zeta_N
\right)
_\ast\colon T(J^1M)\to T(M\times _NC^{\mathrm{sym}})
\]
is the Jacobian mapping induced by $\zeta_N$, and
$\left( \zeta_N^v\right) _\ast\colon V(p_{M}^1)\to V(p)$
is its restriction to the vertical bundles.
\end{theorem}

\begin{proof}
As a computation shows, the equations of
$\zeta_N$ in the coordinate systems introduced
in the section \ref{Coordinate Systems}, are as follows:
\begin{align}
x^i\circ\zeta_N
& =x^i,
\nonumber\\
y_{ij}\circ\zeta_N
& =y_{ij},
\nonumber\\
A_{ij}^h\circ\zeta_N
& =\tfrac{1}{2}y^{hk}(y_{ik,j}+y_{jk,i}-y_{ij,k}),
\quad i\leq j,
\label{A_circ_zeta}
\end{align}
where $(y^{ij})_{i,j=1}^n$ is the inverse mapping of the matrix
$(y_{ij})_{i,j=1}^n$ and the functions $y_{ij}$ are defined in
\eqref{inducedsystemM(N)}. Hence
\begin{align}
x^i\circ\zeta_N^{-1}  & =x^i,
\nonumber\\
y_{ij}\circ\zeta_N^{-1}  & =y_{ij},
\nonumber\\
y_{ij,k}\circ\zeta_N^{-1}  & =y_{hi}A_{jk}^h+y_{hj}A_{ik}^h,
\quad i\leq j.
\label{y_circ_zeta^-1}
\end{align}
As the diffeomorphism $\zeta_N$ induces the identity
on the ground manifold $N$, it follows that the definition
of $\gamma^2$ in \eqref{gamma2} makes
sense and the following formulas are obtained:
\[
\gamma^2
\left(
\frac{\partial}{\partial x^r}
\right)
=\sum_{a\leq b}
\left(
\gamma _{abr}\circ\zeta_N
\right)
\frac{\partial}{\partial y_{ab}}
+\sum_{i\leq j}\gamma _{ijkr}
\frac{\partial}{\partial y_{ij,k}},
\]
\begin{align*}
\gamma _{ijkr}
& =\tfrac{1}{2}\sum _{a\leq b}
\tfrac{\delta _{ah}\delta _{bi}+\delta _{ai}\delta _{bh}}
{1+\delta _{hi}}
\left(
\gamma _{abr}\circ \zeta _N
\right)
y^{hl}(y_{jl,k}+y_{kl,j}-y_{jk,l})\\
& +\tfrac{1}{2}\sum _{a\leq b}
\tfrac{\delta _{ah}\delta _{bj}+\delta _{aj}\delta _{bh}}
{1+\delta _{hj}}
\left(
\gamma _{abr}\circ \zeta _N
\right)
y^{hl}(y_{il,k}+y_{kl,i}-y_{ik,l})\\
& +\sum _{j\leq a}
\tfrac{\delta _{ak}}{1+\delta _{jk}}
\left(
\gamma _{jar}^h\circ \zeta_N
\right)
y_{hi}
+\sum _{a\leq j}\tfrac{\delta _{ak}}{1+\delta _{jk}}
\left(
\gamma _{ajr}^h\circ \zeta_N
\right)
y_{hi}\\
& +\sum _{i\leq a}\tfrac{\delta _{ak}}{1+\delta _{ik}}
\left(
\gamma _{iar}^h\circ \zeta_N
\right)
y_{hj}+\sum _{a\leq i}\tfrac{\delta _{ak}}{1+\delta _{ik}}
\left(
\gamma _{air}^h\circ \zeta_N
\right)
y_{hj},
\end{align*}
where
\[
\gamma=\sum_{i\leq j}
\left(
dy_{ij}+\gamma _{ijk}dx^k
\right)
\otimes \dfrac{\partial}{\partial y_{ij}}
+\sum_{j\leq k}
\left(
dA_{jk}^i+\gamma _{jkl}^idx^l
\right)
\otimes\dfrac{\partial}{\partial A_{jk}^i},
\]
or equivalently,
\[
\gamma=\tfrac{1}{2-\delta_{ij}}\left(  dy_{ij}+\gamma _{ijk}dx^k\right)
\otimes\dfrac{\partial}{\partial y_{ij}}+\tfrac{1}{2-\delta_{jk}}\left(
dA_{jk}^i+\gamma _{jkl}^idx^l\right)  \otimes\dfrac{\partial}{\partial
A_{jk}^i},
\]
assuming $\gamma _{hir}=\gamma _{ihr}$ for $h>i$, and
$\gamma _{jkr}^h=\gamma _{kjr}^h$ for $j>k$.
Taking the symmetry $A_{jk}^i=A_{kj}^i$ into account, we obtain
\begin{align*}
\gamma _{ijkr}  & =\tfrac{1}{2}
\left(
\gamma _{hir}\circ\zeta_N
\right)
y^{hl}(y_{jl,k}+y_{kl,j}-y_{jk,l})\\
& +\tfrac{1}{2}
\left(
\gamma _{hjr}\circ\zeta_N
\right)
y^{hl}(y_{il,k}+y_{kl,i}-y_{ik,l})\\
& +
\left(
\gamma _{jkr}^h\circ\zeta_N
\right)  y_{hi}
+\left(
\gamma _{ikr}^h\circ\zeta_N
\right)  y_{hj}.
\end{align*}
Hence
\begin{equation}
\gamma _{ijkr}\circ\zeta_N^{-1}
=\gamma _{hir}A_{jk}^h
+\gamma _{hjr}A_{ik}^h
+\gamma _{jkr}^hy_{hi}
+\gamma _{ikr}^hy_{hj},\ i\leq j.
\label{gamma_ijkr_zeta^-1}
\end{equation}
Permuting the indices $i,j,k$ cyclically
on the previous equation, we have
\begin{equation}
\gamma _{ijr}^s=-\gamma _{hkr}A_{ij}^hy^{ks}
-\tfrac{1}{2}\left( \gamma _{ijkr}\circ\zeta_N^{-1}
-\gamma _{jkir}\circ\zeta_N^{-1}
-\gamma _{kijr}\circ\zeta_N^{-1}\right)  y^{ks},
\label{gamma^s_ijr}
\end{equation}
thus proving that the mapping $\gamma\mapsto\gamma^2$
defined in the statement, is bijective.
\end{proof}

\subsection{Covariant Hamiltonians for second-order Lagrangians\label{s5}}

The Legendre form of a second-order Lagrangian density
$\Lambda=Lv_n$ on the bundle $p\colon E\to N$
is the $V^\ast (p^1)$-valued $p^3$-horizontal $(n-1)$-form
$\omega_{\Lambda}$ on $J^3E$ locally given by
(e.g., see \cite{Gotay}, \cite{Mu2}, \cite{SCr}),
\[
\omega_{\Lambda}
=i_{\partial/\partial x^i}v_n\otimes
\left(
L_{\alpha }^{i0}dy^\alpha
+L_{\alpha}^{ij}dy_j^\alpha
\right) ,
\]
where
\begin{align}
L_{\alpha}^{ij}
&  =\tfrac{1}{2-\delta_{ij}}
\frac{\partial L}{\partial y_{ij}^\alpha },
\label{f11}\\
L_{\alpha}^i
&  =\frac{\partial L}{\partial y_i^\alpha }
-\sum_j\tfrac{1}{2-\delta_{ij}}D_j
\left(
\frac{\partial L}{\partial y_{ij}^\alpha }
\right)  ,
\label{f12}
\end{align}
and
\[
D_j=\frac{\partial }{\partial x^j}
+\sum_{I\in \mathbb{N}^n,
|I|=0}^\infty y_{I+(j)}^\alpha
\frac{\partial }{\partial y_I^\alpha }
\]
denotes the total derivative with respect
to the variable $x^j$.

The Poincar\'{e}-Cartan form attached to $\Lambda $
is then defined to be the ordinary $n$-form on $J^3E$
given by,
$\Theta _\Lambda
=(p_2^3)^\ast \theta ^2\wedge \omega _\Lambda +\Lambda $,
where $\theta^2$ is the second-order structure form
(cf.\ \cite[(0.36)]{SardaZak}) and the exterior product
of $(p_2^3)^\ast \theta^2$ and the Legendre form,
is taken with respect to the pairing induced by duality,
$V(p^1)\times _{J^1E}V^\ast (p^1)\to\mathbb{R}$.
The most outstanding difference with the
first-order case is that the Legendre and
Poincar\'{e}-Cartan forms associated
with a second-order Lagrangian density
are generally defined on $J^3E$, thus
increasing by one the order of the density.

Similarly to the first-order case
(see \cite{FGM}, \cite{MSh}), given a
second-order Lagrangian density $\Lambda $ on
$p\colon E\to N$ and a second-order connection
$\gamma^2$ on $p\colon E\to N$, by
subtracting $(p_2^3)^\ast \theta^2$ from
$(p_1^3)^\ast \gamma^2$ we obtain a $p^3$-horizontal
form, and we can define the corresponding
covariant Hamiltonian to be the Lagrangian
density $\Lambda^{\gamma^2}$ of third order,
\begin{equation}
\Lambda^{\gamma^2}=
\left(
(p_1^3)^\ast \gamma^2-(p_2^3)^\ast \theta^2
\right)
\wedge\omega_\Lambda-\Lambda.
\label{Lambda_gamma_2}
\end{equation}
Expanding on the right-hand side of the
previous equation, we obtain a decomposition of
$\Theta_\Lambda$ that generalizes the classical
formula for the Hamiltonian in Mechanics;
namely,
$\Theta_\Lambda=(p_1^3)^\ast \gamma^2\wedge
\omega_\Lambda-\Lambda^{\gamma^2}$.
With the same notations as in the formulas
\eqref{f2}, \eqref{f11}, \eqref{f12} the
following formula is deduced:
\begin{equation}
L^{\gamma ^2}
=(\gamma _i^\alpha +y_i^\alpha )L_\alpha^{i0}
+(\gamma _{hi}^\alpha +y_{hi}^\alpha
)L_{\alpha}^{ih}-L.
\label{f19}
\end{equation}
Because of the equation \eqref{f12},
$\Theta _\Lambda $ and $L^{\gamma ^2}$
are generally defined on $J^3E$.

\subsection{Invariant covariant Hamiltonians on $J^2M$}

\begin{lemma}
\label{lemma2}
If $\gamma$ is a first-order Ehresmann connection
on $M\times _NC^{\mathrm{sym}}$ satisfying
the conditions $(C_M)$, then the following equation
holds for the second-order Ehresmann connection
$\gamma ^2$ on $M$ given in the formula
\emph{\eqref{gamma2}}:
\[
\gamma _{abr}\circ \zeta _N=-y_{ab,r}.
\]

\end{lemma}
\begin{proof}
Actually, from the formulas \eqref{gamma_klj}
and \eqref{A_circ_zeta} we obtain
\begin{align*}
\! \gamma _{abr}\circ \zeta _N\!
& =\!-\left(
y_{mb}
\left(
A_{ra}^m\circ \zeta _N
\right)
\!+\!y_{ma}
\left(
A_{rb}^m\circ \zeta _N
\right)
\right) \\
\! & =\!-\tfrac{1}{2}
\left\{
y_{mb}y^{mk}(y_{rk,a}\!+\!y_{ak,r}
\!-\!y_{ra,k})
\!+\! y_{ma}y^{mk}(y_{rk,b}
\!+\! y_{bk,r}
\!-\! y_{rb,k})
\right\} \\
\! & =\!-y_{ab,r}.
\end{align*}

\end{proof}

\begin{lemma}
\label{lemma3}If a first-order connection $\gamma$ on
$M\times _NC^{\mathrm{sym}}$ satisfies the condition $(C_C)$
introduced above, then the
following formulas for its components hold:
\begin{equation}
\gamma _{rts}^h-\gamma _{rst}^h
=A_{sm}^hA_{rt}^m-A_{tm}^hA_{rs}^m.
\label{28bis}
\end{equation}

\end{lemma}

\begin{proof}
As the bundle under consideration is that
of symmetric connections, the following symmetry holds:
$\gamma _{abc}^h=\gamma _{bac}^h$, and we have
\[
\begin{array}
[c]{rl}
\gamma _{rts}^h=\gamma _{str}^h-
\left(
A_{rm}^hA_{st}^m -A_{sm}^hA_{rt}^m
\right)  &
\left[
\text{by virtue of \eqref{28}}
\right] \\
=\gamma _{tsr}^h-
\left(
A_{rm}^hA_{st}^m-A_{sm}^hA_{rt}^m
\right)  &
\\
=\left(
\gamma _{rst}^h+A_{rm}^hA_{st}^m-A_{tm}^hA_{rs}^m
\right)  &
\left[
\text{by virtue of \eqref{28}}
\right] \\
-\left(
A_{rm}^hA_{st}^m-A_{sm}^hA_{rt}^m
\right)  & \\
=\gamma _{rst}^h+
\left(
A_{sm}^hA_{rt}^m-A_{tm}^hA_{rs}^m
\right)  &
\end{array}
\]
\end{proof}

\begin{proposition}
\label{retract}Let
\[
\zeta_N^2=
\left.
\zeta_N^{(1)}
\right
\vert _{J^2M}\colon J^2M\to J^1(M\times _NC^{\mathrm{sym}})
\]
be the restriction to the closed submanifold
$J^2M\subset J^1(J^1M)$ of the prolongation
$\zeta_N^{(1)}\colon J^1(J^1M)\to J^1(M\times _NC^{\mathrm{sym}})$
of the mapping $\zeta_N$ defined in
\emph{Theorem \ref{one-to-one}}. For every
$(j_x^1g,j_x^1\Gamma)\in  J^1(M\times _NC^{\mathrm{sym}})$
there exists a unique $j_x^2g^{\prime }\in  J_x^2M$ such that,
$j_x^1g^\prime =j_x^1g$ and
$j_x^1\Gamma^{g^\prime }=j_x^1\Gamma$ and the mapping
$\varkappa\colon J^1(M\times _NC^{\mathrm{sym}})\to J^2M$
defined by
$\varkappa(j_x^1g,j_x^1\Gamma)=j_x^2g^\prime $ is a
$\mathrm{Diff}N$-equivariant rectract of $\zeta_N^2$.
\end{proposition}

\begin{proof}
From the formulas \eqref{A_circ_zeta} and
\eqref{y_circ_zeta^-1} we obtain
\begin{align*}
\frac{\partial g_{ij}^\prime }{\partial x^k}
& =g_{hi}^\prime
\left(
\Gamma^{g^\prime }
\right)  _{jk}^h+g_{hj}^\prime
\left(
\Gamma ^{g^\prime }
\right)
_{ik}^h,\\
\left(
\Gamma^{g^\prime }
\right)  _{ij}^h
& =\tfrac{1}{2}g^{\prime hk}
\left(
\frac{\partial g_{ik}^\prime }{\partial x^j}
+\frac{\partial g_{jk}^\prime }{\partial x^i}
-\frac{\partial g_{ij}^\prime }{\partial x^k}
\right)
\end{align*}
for every non-singular metric $g^\prime $ on $N$.
Hence the second partial derivatives of $g_{ij}^\prime $
are completely determined, namely
\[
\frac{\partial^2g_{ij}^\prime }{\partial x^k\partial x^l}
=\frac{\partial g_{hi}}{\partial x^l}\Gamma _{jk}^h
+g_{hi}\frac{\partial\Gamma _{jk}^h}{\partial x^l}
+\frac{\partial g_{hj}}{\partial x^l}\Gamma _{ik}^h
+g_{hj}\frac{\partial\Gamma _{ik}^h}{\partial x^l}.
\]

Moreover, the Levi-Civita connection of a metric
depends functorially on the metric, i.e.,
$\phi\cdot\Gamma^g=\Gamma^{\phi\cdot g}$ for every
$\phi \in \mathrm{Diff}N$. Hence, by transforming
the equations $j_x^1g^\prime =j_x^1g$ and
$j_x^1\Gamma^{g^\prime }=j_x^1\Gamma^g$ by $\phi$
we can conclude.
\end{proof}

\begin{theorem}
\label{Theorem for P2}If a first-order Ehresmann
connection $\gamma$ on $M\times _NC^{\mathrm{sym}}$
satisfies the conditions $(C_M)$ and $(C_C)$ introduced
above, then the covariant Hamiltonian $\Lambda ^{\gamma ^2}$
attached to every $\mathrm{Diff}N$-invariant second-order
Lagrangian density $\Lambda $ on $M$ with respect
to the second-order Ehresmann connection $\gamma^2$
on $M$ defined in the formula \emph{\eqref{gamma2}},
is defined on $J^2M$ and it is also $\mathrm{Diff}N$-invariant.
\end{theorem}

\begin{proof}
Given a $\mathrm{Diff}N$-invariant
second-order Lagrangian density
$\Lambda=\mathcal{L}\mathbf{v}$ on $M$, let
$\Lambda^\prime =\mathcal{L}^\prime \mathbf{v}$
be the first-order Lagrangian density on
$M\times _NC^{\mathrm{sym}}$ given by
$\Lambda^\prime =\varkappa^\ast \Lambda$,
which is also $\mathrm{Diff}N$-invariant as
$\varkappa$ is a $\mathrm{Diff}N$-equivariant mapping
according to Proposition \ref{retract}.
Moreover, as $\varkappa$ is a retract of $\zeta_N^2$,
we have
$\left(
\zeta _N^2
\right) ^\ast \Lambda^\prime
=\left(
\zeta_N^2
\right)
^\ast \varkappa^\ast \Lambda
=(\varkappa\circ\zeta_N^2)^\ast \Lambda=\Lambda$,
i.e.,
$\Lambda=\left(
 \zeta_N^2
 \right) ^\ast \Lambda^\prime $.
This formula is equivalent to saying
$\mathcal{L}=\mathcal{L}^\prime \circ\zeta_N^2$, as the
$n$-form $\mathbf{v}$ is $\mathrm{Diff}N$-invariant,
and it is even equivalent to
$L=L^\prime \circ\zeta_N^2$ because
$\zeta _N^2$ induces the identity on $N$.

We claim $\mathcal{L}^{\gamma^2}=
\left( \mathcal{L}^\prime
\right) ^\gamma \circ\zeta_N^2$.
This formula will end the proof as the mapping
$\zeta_N^2$ is $\mathrm{Diff}N$-equivariant and
$\left(  \mathcal{L}^\prime \right) ^\gamma $ is
$\mathrm{Diff}N$-invariant by virtue of Theorem \ref{th1}.

To start with, we observe that the formula \eqref{f11}
for $\Lambda $ can be written, in the present case,
as follows:
\[
L^{abij}=\tfrac{1}{2-\delta_{ij}}
\frac{\partial L}{\partial y_{ab,ij}},
\]
or equivalently, letting
$\mathcal{L}^{abij}=\rho^{-1}L^{abij}$,
\begin{equation}
\mathcal{L}^{abij}=\tfrac{1}{2-\delta_{ij}}
\frac{\partial\mathcal{L}}{\partial y_{ab,ij}}.
\label{L^abij}
\end{equation}
Taking the formula in Lemma \ref{lemma2}
into account, the formula \eqref{f19}
for $\Lambda $ reads as
$L^{\gamma ^2}=\sum _{a\leq b}(\gamma _{abij}
+y_{ab,ij})L^{abij}-L$,
or even
\[
\mathcal{L}^{\gamma^2}
=\sum_{a\leq b}(\gamma _{abij}
+y_{ab,ij})\mathcal{L}^{abij}-\mathcal{L},
\]
where
$\mathcal{L}^{\gamma ^2}=\rho ^{-1}L^{\gamma ^2}$.
Hence $\mathcal{L}^{\gamma ^2}$ is defined over
$J^2M$. As $y_{ab,ij}=y_{ab,ji}$, we obtain
\begin{multline*}
\mathcal{L}^{\gamma ^2}
\!=\!
\sum_{a\leq b}\sum_{i\leq j}
\left(
\tfrac{1}{2}
\left(
\gamma _{abij}+\gamma _{abji}
\right)
+y_{ab,ij}
\right)
\frac{\partial \left( \mathcal{L}^\prime \circ\zeta_N^2\right) }
{\partial y_{ab,ij}}
-\mathcal{L}^\prime \circ \zeta _N^2\\
\!=\! \sum _{a\leq b}\sum_{i\leq j}\sum_{k\leq l}
\left(
\tfrac{1}{2}
\left(
\gamma _{abij}+\gamma _{abji}
\right)
+y_{ab,ij}
\right)
\left(
\!\frac{\partial\mathcal{L}^\prime }
{\partial A_{kl,q}^h}
\circ\zeta _N^2\!
\right)
\frac{\partial(A_{kl,q}^h\circ \zeta _N^2)}
{\partial y_{ab,ij}}
-\mathcal{L}^\prime \circ \zeta _N^2\\
=\sum_{k\leq l}\tfrac{1}{4}y^{hm}
\left(
\gamma _{kmql}+\gamma _{kmlq}
+\gamma _{lmqk}+\gamma _{lmkq}
-\gamma _{klqm}-\gamma _{klmq}
\right)
\left(
\frac{\partial\mathcal{L}^\prime }
{\partial A_{kl,q}^h}
\circ \zeta _N^2
\right) \\
+\sum_{k\leq l}
\tfrac{1}{2}y^{hm}
\left(
y_{km,ql}+y_{lm,qk}-y_{kl,qm}
\right)
\left(
\frac{\partial\mathcal{L}^\prime }
{\partial A_{kl,q}^h}
\circ \zeta_N^2
\right)
-\mathcal{L}^\prime \circ \zeta _N^2.
\end{multline*}
\newline Moreover, we have
\[
\left(
\mathcal{L}^\prime
\right) ^\gamma
=\sum_{a\leq b}
\left(
\gamma _{abc}+y_{ab,c}
\right)
\frac{\partial \mathcal{L}^\prime }
{\partial y_{ab,c}}
+\sum_{a\leq b}
\left(
\gamma _{abl}^i+A_{ab,l}^i
\right)
\frac{\partial \mathcal{L}^\prime }
{\partial A_{ab,l}^i}
-\mathcal{L}^\prime .
\]
Hence
\begin{align*}
\left(
\mathcal{L}^\prime
\right)  ^\gamma \circ\zeta_N^2
&=\sum_{k\leq l}
\left(
\gamma _{klq}^h\circ\zeta_N
+A_{kl,q}^h\circ \zeta_N
\right)
\left(
\frac{\partial\mathcal{L}^\prime }
{\partial A_{kl,q}^h}
\circ \zeta _N^2
\right)
-\mathcal{L}^\prime \circ\zeta
_N^2\\
& =\sum_{k\leq l}
\left\{
-\tfrac{1}{2}
\left(
\gamma _{klrq}-\gamma _{lrkq}
-\gamma _{rklq}
\right)
y^{rh}
\right. \\
& \left.
+\tfrac{1}{2}
\left(
y_{kr,lq}+y_{lr,kq}-y_{kl,rq}
\right)
y^{hr}
\right\}
\left(
\frac{\partial\mathcal{L}^\prime }
{\partial A_{kl,q}^h}
\circ \zeta _N^2
\right)
-\mathcal{L}^\prime \circ \zeta _N^2.
\end{align*}

Consequently, the proof reduces to state
that the following equation
\[
\tfrac{1}{4}
\left(
\gamma _{krql}+\gamma _{krlq}+\gamma _{lrqk}
+\gamma _{lrkq}-\gamma _{klqr}-\gamma _{klrq}
\right)
=-\tfrac{1}{2}
\left(
\gamma _{klrq}-\gamma _{lrkq}-\gamma _{rklq}
\right)
\]
holds true, or equivalently,
\begin{equation}
0=\left(
\gamma _{ijkr}-\gamma _{ijrk}
\right)
+\left(
\gamma _{irjk}-\gamma _{irkj}
\right)
+\left(
\gamma _{rjki}-\gamma _{rjik}
\right) .\label{equation_bis}
\end{equation}
According to the formulas
\eqref{gamma_ijkr_zeta^-1}
and \eqref{gamma_klj} we obtain
\begin{align*}
\gamma_{ijkr}\circ \zeta _N^{-1}
& =\left(
\gamma _{jkr}^h-A_{ra}^{h}A_{jk}^a
\right)
y_{hi}
+\left(
\gamma _{ikr}^h-A_{ra}^{h}A_{ik}^a
\right)
y_{hj} \\
& -\left(
A_{rj}^hA_{ik}^a+A_{ri}^{h}A_{jk}^a
\right)  y_{ah}.
\end{align*}
The third term on the right-hand side
of this equation is symmetric
in the indices $k$ and $r$,
as $A_{bc}^a=A_{cb}^a$. Hence
\begin{align*}
\left(
\gamma _{ijkr}-\gamma _{ijrk}
\right)
\circ\zeta_N^{-1}
& =\left(
\gamma _{jkr}^h
-\gamma _{jrk}^h
-A_{ra}^hA_{jk}^a
+A_{ka}^hA_{jr}^a
\right)
y_{hi}\\
& +
\left(
\gamma _{ikr}^h-\gamma _{irk}^h
-A_{ra}^hA_{ik}^a +A_{ka}^hA_{ir}^a
\right)
y_{hj}.
\end{align*}
By composing the right-hand side of the equation
\eqref{equation_bis} and $\zeta _N^{-1}$,
and taking the previous formula and the formulas
\eqref{28} and \eqref{28bis} into account,
we conclude that this expression vanishes indeed.
\end{proof}

\section{Palatini and Einstein-Hilbert Lagrangians}

Let us compute the covariant Hamiltonian
density attached to the Palatini Lagrangian.
Following the notations in \cite{KN},
the Ricci tensor field attached to
the symmetric connection $\Gamma$ is given by
$S^\Gamma(X,Y)
=\operatorname*{tr}(Z\mapsto R^\Gamma (Z,X)Y)$,
where $R^\Gamma $ denotes the curvature tensor
field of the covariant derivative
$\nabla ^\Gamma $ associated to $\Gamma $
on the tangent bundle; hence
$S^\Gamma=(R^\Gamma )_{jl}dx^l\otimes dx^j$,
where
\begin{align*}
(R^\Gamma )_{jl}
& =(R^\Gamma )_{jkl}^k,\\
(R^\Gamma )_{jkl}^i
& =\partial \Gamma _{jl}^i/\partial x^k
-\partial \Gamma _{jk}^i/\partial x^l
+\Gamma _{jl}^m\Gamma _{km}^i
-\Gamma _{jk}^m\Gamma _{lm}^i.
\end{align*}
The Lagrangian is the function on
$J^1(M\times _NC^{\mathrm{sym}})$ thus given by,
\[
\mathcal{L}_P(j_x^1g,j_x^1\Gamma )
=g^{ij}(x)(R^\Gamma )_{ij}(x)
\]
and local expression
\[
\mathcal{L}_{P}=y^{ij}(A_{ij,k}^k-A_{ik,j}^k
+A_{ij}^mA_{km}^k-A_{ik}^mA_{jm}^k) .
\]
As a computation shows, for every first-order
connection $\gamma $ on $M\times _NC^{\mathrm{sym}}$
satisfying \eqref{28bis} and taking the formula
\eqref{Lambdagamma1} into account, we obtain
$\mathcal{L}_P^\gamma =0$. This result is essentially
due to the fact that the P-C form of the P density
$\Lambda_P=\mathcal{L}_P\mathbf{v}=L_Pv_n$
projects onto $M\times _NC^{\mathrm{sym}}$. In fact,
the following general characterization holds:

\begin{proposition}
Let $p\colon E\to N$ be an arbitrary fibred
manifold and let $\gamma$ be a first-order
Ehresmann connection on $E$. The equation
$L^\gamma =0$ holds true for a Lagrangian
$L\in  C^\infty (J^1E)$ if and only if,
\emph{i)} the Poincar\'{e}-Cartan form
of the density $\Lambda=Lv_n$ projects onto
$J^{0}E$ and,
\emph{ii)}
$L=\left\langle
(p_0^1)^\ast \gamma-\theta,dL|_{V(p_0^1)}
\right\rangle $.
\end{proposition}

\begin{proof}
The equation $L^\gamma =0$ is equivalent
to the equation $D^\gamma L=L$, where $D^\gamma $
is the $p_0^1$-vertical vector field defined in the
formula \eqref{Dgamma}, and the general solution
to the latter is
$L=f(x^i,y^\alpha ,\gamma _i^\alpha +y_i^\alpha )$,
$f(x^i,y^\alpha ,y_i^\alpha )$ being a homogeneous
smooth function of degree one in the variables
$(y_i^\alpha )$, $1\leq \alpha \leq m$, $1\leq i\leq n$,
according to Euler's homogeneous function theorem.
As $f$ is defined for all values of the variables
$(y_i^\alpha )$, $1\leq \alpha \leq m$, $1\leq i\leq n$,
we conclude that the functions
$L_\alpha ^i=\partial L/\partial y_i^\alpha $
must be defined on $E$. Hence $L$ is written as
$L=L_\alpha ^i(x^j,y^\beta )y_i^\alpha +L_0(x^j,y^\beta )$,
but this is exactly the condition for the P-C form of
$\Lambda $ to be projectable onto $J^0E=E$,
as follows from the local expression of this form, namely,
\begin{align*}
\Theta _\Lambda
& =\frac{\partial L}{\partial y_i^\alpha }
\theta ^\alpha \wedge i_{\partial/\partial x^i}v_n
+Lv_n\\
& =\frac{\partial L}{\partial y_i^\alpha }dy^\alpha
\wedge i_{\partial /\partial x^i}v_n
+\left(
L-y_i^\alpha
\frac{\partial L}{\partial y_i^\alpha }
\right)
v_n.
\end{align*}

Moreover, by imposing the condition
$D^\gamma L=L$ we obtain
$L_0=L_\alpha^i\gamma _i^\alpha $,
or in other words
$L=(\gamma _i^\alpha+y_i^\alpha )
\partial L/\partial y_i^\alpha $,
which is equivalent to the equation ii)
in the statement.
\end{proof}

The corresponding result for the second-order
formalism is similar but the computations
are more cumbersome. Let us compute the covariant
Hamiltonian density attached to the Einstein-Hilbert
Lagrangian. As a matter of notation, we set
$S^g(X,Y)=S^{\Gamma^g}(X,Y)$ for the metric $g$,
$\Gamma^g$ being its Levi-Civita connection,
and similarly, $(R^g)_{jkl}^i=(R^{\Gamma^g})_{jkl}^i$.

The E-H Lagrangian is thus given by
$\mathcal{L}_{EH}\circ j^2g =(y^{ij}\circ g)(R^g)_{ihj}^h$.
As the Levi-Civita connection $\Gamma ^g$ depends
functorially on $g$, $\mathcal{L}_{EH}$
is readily seen to be $\mathrm{Diff}N$-invariant;
it is in addition linear in the second-order variables
$y_{ij,kl}$. By using the third formula in
\eqref{A_circ_zeta} the following local expression
for $\mathcal{L}_{EH}$ is obtained:
\[
\mathcal{L}_{EH}=\tfrac{1}{2}y^{ij}y^{hd}
\left(
y_{dj,hi}-y_{ij,dh}
-y_{dh,ij}+y_{hi,dj}
\right)
+\mathcal{L}_{EH}^\prime ,
\]
\begin{align*}
\mathcal{L}_{EH}^\prime
& =\tfrac{1}{2}y^{ij}
\left\{
y^{hm}y_{mr,j}y^{rd}
\left(
y_{id,h}+y_{hd,i}-y_{ih,d}
\right)
\right. \\
& -y^{hm}y_{mr,h}y^{rd}\left(
y_{id,j}+y_{jd,i}-y_{ij,d}\right) \\
& +\tfrac{1}{2}y^{hr}y^{md}
\left(
y_{id,j}+y_{jd,i}-y_{ij,d}
\right)
\left(
y_{hr,m}+y_{mr,h}-y_{hm,r}
\right) \\
& \left.
-\tfrac{1}{2}y^{hr}y^{md}
\left(
y_{id,h}+y_{hd,i}-y_{ih,d}
\right)
\left(
y_{jr,m}+y_{mr,j}-y_{jm,r}
\right)
\right\} .
\end{align*}
According to \eqref{L^abij}, for every first-order
connection form $\gamma $ on $M\times _NC^{\mathrm{sym}}$
satisfying the conditions $(C_M)$ and $(C_C)$ above,
we have
\[
\mathcal{L}_{EH}^{\gamma^2}
=\sum_{a\leq b}\tfrac{1}{2-\delta_{ij}}
(\gamma _{abij}+y_{ab,ij})
\frac{\partial\mathcal{L}_{EH}}{\partial y_{ab,ij}}
-\mathcal{L}_{EH},
\]
and as a computation shows,
\begin{align*}
\mathcal{L}_{EH}^{\gamma^2}
& =\tfrac{1}{2}y^{ij}
\left(
\gamma _{idjh}
+\gamma _{jdih}
-\gamma _{ijdh}
-\gamma _{idhj}
-\gamma _{hdij}
+\gamma _{ihdj}
\right)  y^{hd}\\
& +\tfrac{1}{2}y^{ij}
\left\{
y^{hm}y_{mr,h}y^{rd}
\left(
y_{id,j}+y_{jd,i}-y_{ij,d}
\right)
\right. \\
& -y^{hm}y_{mr,j}y^{rd}
\left(
y_{id,h}+y_{hd,i}-y_{ih,d}
\right) \\
& -\tfrac{1}{2}y^{hr}y^{md}
\left(
y_{id,j}+y_{jd,i}-y_{ij,d}
\right)
\left(
y_{hr,m}+y_{mr,h}-y_{hm,r}
\right) \\
& \left.
+\tfrac{1}{2}y^{hr}y^{md}
\left(
y_{id,h}+y_{hd,i}-y_{ih,d}
\right)
\left(
y_{jr,m}+y_{mr,j}-y_{jm,r}
\right)
\right\} \\
& =0,
\end{align*}
where the formulas \eqref{gamma^s_ijr},
\eqref{28bis}, \eqref{A_circ_zeta},
and Lemma \ref{lemma3} have been used.
In this case, the P-C form of the E-H
density
$\Lambda_{EH}=\mathcal{L}_{EH}\mathbf{v}=L_{EH}v_n$,
\begin{align}
\Theta_{\Lambda_{EH}}
& =\sum\nolimits_{k\leq l}
\left(
 L_{EH}^{i,kl}dy_{kl}+L_{EH}^{ij,kl}dy_{kl,j}
\right)
\wedge i_{\partial/\partial x^i}v_n+Hv_n,
\label{Theta_Lambda_EH}\\
H  & =L_{EH}^\prime
-\sum\nolimits_{k\leq l}L_{EH}^{i,kl}y_{kl,i},
\nonumber\\
L_{EH}^{i,kl}
& =\frac{\partial L_{EH}^\prime }{\partial y_{kl,i}}
-\tfrac{1}{2-\delta _{ij}}y_{ab,j}
\frac{\partial^2L_{EH}}
{\partial y_{ab}\partial y_{kl,ij}},
\nonumber\\
L_{EH}^{ij,kl}
& =\tfrac{1}{2-\delta_{ij}}
\frac{\partial L_{EH}}{\partial y_{kl,ij}},
\nonumber
\end{align}
(cf.\ \eqref{f11}, \eqref{f12}) is not only
projectable onto $J^2M$ but also on $J^1M$
(e.g., see \cite{GM1}), although there is
no first-order Lagrangian on $J^1M$ admitting
\eqref{Theta_Lambda_EH} as its P-C form.
This fact is strongly related to a classical
result by Hermann Weyl (\cite[Appendix II]{Weyl},
also see \cite{Lovelock}, \cite{Heyde}) according
to which the only $\mathrm{Diff}N$-invariant
Lagrangians on $J^2M$ depending linearly
on the second-order coordinates
$y_{ab,ij}$ are of the form
$\lambda\mathcal{L}_{EH}+\mu$,
for scalars $\lambda $, $\mu $.
This also explains why a true first-order
Hamiltonian formalism exists in the Einstein-Cartan
gravitation theory, e.g., see \cite{Szczyrba1},
\cite{Szczyrba2}. In fact, if
\[
L_{EH}^i=\tfrac{1}{2-\delta_{ij}}
\frac{\partial L_{EH}}{\partial y_{kl,ij}}
y_{kl,j}
\quad
\left(
\text{hence \ }L_{EH}^{ij,kl}
=\frac{\partial L_{EH}^i}{\partial y_{kl,j}}
\right)
\]
and the momentum functions are defined as follows:
\[
p_{kl,i}=L_{EH}^{i,kl}-
\frac{\partial L_{EH}^i}{\partial y_{kl}},
\]
then
\[
d\Theta_{\Lambda_{EH}}
=dp_{kl,i}\wedge dy_{kl}\wedge
i_{\partial/\partial x^i}v_n
+dH\wedge v_n,
\]
and from the Hamilton-Cartan equation
(e.g., see \cite[(1)]{GM1}) we conclude
that a metric $g$ is an extremal for
$\Lambda_{EH}$ if and only if,
\begin{align*}
0  & =\frac{\partial(p_{ab,i}\circ j^1g)}
{\partial x^i}
-\frac{\partial H}{\partial y_{ab}}\circ j^1g,
\\
0  & =\frac{\partial(y_{ab}\circ g)}{\partial x^i}
+\frac{\partial H}{\partial y_{ab,i}}\circ j^1g.
\end{align*}

On the other hand, it is no longer true that the
covariant Hamiltonians of the non-linear Lagrangians
of the form $f(\mathcal{L}_{EH})$,
$f^{\prime\prime }\neq 0$, considered
in some cosmological models (e.g., see
\cite{BorowiecFerrarisFrancavigliaVolovich},
\cite{CotsakisMiritzisQuerella}, \cite{DuruisseauKerner},
\cite{Flanagan}, \cite{Kerner}, \cite{KoivistoKurki},
\cite{Poplawski}) and those in higher dimensions (e.g.,
see \cite{GiorginiKerner}, \cite{Shahid-Saless}) vanish.
In fact, as a computation shows, one has
$f(\mathcal{L}_{EH}\mathcal{)}^{\gamma^2}
=f^\prime (\mathcal{L}_{EH})\mathcal{L}_{EH}
-f(\mathcal{L}_{EH})$,
$\forall f\in  C^\infty (\mathbb{R})$.

\bibliographystyle{my-h-elsevier}

\end{document}